\documentclass[11pt]{article}

\usepackage{amsfonts,amsmath,amssymb, bbm, bm, xspace, mathrsfs}
\usepackage{multirow,graphicx, epsfig,subfigure,algorithm,algorithmic,rotating}
\usepackage{cite,url}
\usepackage{fullpage}
\usepackage[small,bf]{caption}
\usepackage[table]{xcolor}
\usepackage{tabularx}
\usepackage{psfrag}
\usepackage{subfigure}
\usepackage{verbatim}
\newtheorem{theorem}{Theorem}
\newtheorem{proposition}[theorem]{Proposition}

\newtheorem{lemma}[theorem]{Lemma}

\newtheorem{corollary}[theorem]{Corollary}
\newtheorem{definition}[theorem]{Definition}
\newtheorem{example}[theorem]{Example}

\def \endprf{\hfill {\vrule height6pt width6pt depth0pt}\medskip}
\newenvironment{proof}{\noindent {\bf Proof} }{\endprf\par}

\numberwithin{equation}{section}

\newcommand{\codebook}{\mathcal{C}} 
\newcommand{\conv}{\operatorname{conv}}
\newcommand{\fundpoly}{\mathcal{P}}
\newcommand{\fundcone}{\mathcal{K}}
\newcommand{\bfxpseudo}{\bfx^{(p)}} 
\newcommand{\xpseudo}{x^{(p)}} 
\newcommand{\lmbpara}{\lambda}
\newcommand{\penpara}{\mu}
\newcommand{\sgn}{\operatorname{sgn}}

\newcommand{\inst}{\operatorname{inst}}

\newcommand{\argmax}{\operatornamewithlimits{argmax}}
\newcommand{\argmin}{\operatornamewithlimits{argmin}}

\newcommand{\Ne}{\mathcal{N}}
\newcommand{\PP}{\mathbb{PP}} 
\newcommand{\Proj}{\Pi}

\newcommand{\SNR}{\mathrm{SNR}}

\newcommand{\mbf}{\boldsymbol}

\newcommand{\bfH}{\mbf{H}}
\newcommand{\bfP}{\mbf{P}}

\newcommand{\bfB}{\mbf{B}}
\newcommand{\bfzero}{\mathbf{0}}

\newcommand{\bfomega}{\mbf{\omega}}
\newcommand{\bfgamma}{\mbf{\gamma}}
\newcommand{\bflambda}{\mbf{\lambda}}

\newcommand{\eye}{\mbf{I}}

\newcommand{\bfa}{\mbf{a}}

\newcommand{\bfc}{\mbf{c}}

\newcommand{\bfn}{\mbf{n}}

\newcommand{\bfu}{\mbf{u}}

\newcommand{\bfv}{\mbf{v}}

\newcommand{\bfw}{\mbf{w}}
\newcommand{\bfx}{\mbf{x}}

\newcommand{\bfX}{\mbf{X}}

\newcommand{\bfy}{\mbf{y}}
\newcommand{\bfY}{\mbf{Y}}
\newcommand{\bfz}{\mbf{z}}

\newcommand{\mcR}{\mathsf{R}}
\newcommand{\mcZ}{\mathsf{Z}}
\newcommand{\mcS}{\mathcal{S}}
\newcommand{\mcB}{\mathcal{B}}
\newcommand{\mcL}{\mathcal{L}}
\newcommand{\mcT}{\mathcal{T}}
\renewcommand{\lmbpara}{\bflambda}
\newcommand{\degi}{|\Ne_v(i)|}

\newcommand{\w}{w}
\newcommand{\bfww}{\mbf{\w}}
\newcommand{\cI}{\mathcal{I}}
\newcommand{\cJ}{\mathcal{J}}

\newcommand{\inputsp}{\mathcal{X}}
\newcommand{\outputsp}{\Sigma}
\newcommand{\bfllr}{\bfgamma}
\newcommand{\llr}{\gamma}
\renewcommand{\min}{\operatornamewithlimits{minimize}}
\newcommand{\suchthat}{\operatorname{subject~to}}
\newcommand{\blocklength}{N}
\newcommand{\checknumber}{M}

\newcommand{\C}{\theta}
\newcommand{\F}{\mathcal{F}}

\newcommand{\pLP}{LP decoding\xspace}
\newcommand{\pPD}{\textbf{ADMM-PD}\xspace}
\newcommand{\pPLP}{\textbf{L1PD}\xspace}
\newcommand{\pRLP}{\textbf{RLPD}\xspace}
\newcommand{\pRLPKDHVB}{\textbf{RLPD-KDHVB}\xspace}

\newcommand{\isapd}{\textbf{ISA-PD}\xspace}
\newcommand{\isarf}{\textbf{ISA-R}\xspace}

\newcommand{\error}{\mathrm{error}}
\newcommand{\new}{\mathrm{new}}
\newcommand{\dB}{\mathrm{dB}}
\renewcommand{\Pr}{P}

\begin{document}
\title{The ADMM penalized decoder for LDPC codes\thanks{This material was presented in part at the 2012 Allerton Conference on Communication, Control, and Computing, Monticello, IL, Sept., 2012. 
This material was also presented in part at the 2012 Information Theory Workshop (ITW), Lausanne, Switzerland, Sept., 2012. This material was also presented in part at the 2014 International Symposium on Information Theory, Honolulu, HI, July, 2014. This work was supported by the National Science
Foundation (NSF) under Grants CCF-1217058 and by a Natural Sciences and Engineering Research Council of Canada (NSERC) Discovery Research Grant. This paper was submitted to \emph{IEEE Trans. Inf. Theory.}}}

\author{Xishuo Liu
\thanks{X.~Liu is with the Dept.~of Electrical and Computer
  Engineering, University of Wisconsin, Madison, WI 53706
  (e-mail: xliu94@wisc.edu).}  
, Stark C. Draper
\thanks{S.~C.~Draper is with the Dept.~of Electrical and Computer
  Engineering, University of Toronto, ON M5S 3G4, Canada (e-mail: stark.draper@utoronto.ca).}
}

\maketitle

\begin{abstract}
Linear programming (LP) decoding for low-density parity-check (LDPC)
codes proposed by Feldman \textit{et al.} is shown to have theoretical guarantees in
several regimes and empirically is not observed to suffer from an error floor. 
However at low signal-to-noise ratios (SNRs), LP decoding is observed to have worse error
performance than belief propagation (BP) decoding. In this paper, we seek to improve 
LP decoding at low SNRs while still achieving good high SNR performance. We first
present a new decoding framework obtained by trying to solve a
non-convex optimization problem using the alternating direction method
of multipliers (ADMM). This non-convex problem is constructed by
adding a penalty term to the LP decoding objective. The goal
of the penalty term is to make ``pseudocodewords'', which are the
non-integer vertices of the LP relaxation to which the LP decoder
fails, more costly.  We name this decoder class the ``ADMM penalized decoder''. 
In our simulation results, the ADMM penalized decoder with $\ell_1$ and $\ell_2$ penalties
outperforms both BP and LP decoding at all SNRs. For high SNR regimes where it is
infeasible to simulate, we use an instanton analysis and show that the ADMM penalized decoder has better high SNR
performance than BP decoding. We also develop a reweighted LP decoder using linear approximations to
the objective with an $\ell_1$ penalty. We show that this decoder has an improved theoretical recovery threshold compared to LP decoding. In addition, we show that the empirical gain of the reweighted LP decoder is significant at low SNRs.
\end{abstract}

\section{Introduction}
\label{sec.introduction}
Low-density parity-check (LDPC) codes, when paired with message
passing decoding algorithms, can perform close to the Shannon limit.
However, message passing is not the only choice for decoding LDPC
codes. The linear programming (LP) decoder proposed by Feldman
\textit{et al}.~\cite{feldman2005using} brings LDPC decoding to the
realm of convex optimization. It has been shown that LP decoding is
better than message passing algorithms in many aspects. First, LP
decoding has an ``ML certificate'' property~\cite{feldman2005using}. 
That is, if the solution
from an LP decoder is an integral solution then it is the maximum
likelihood (ML) solution. Second, LP decoding provides better
theoretical guarantees. Message passing algorithms are not guaranteed
to converge and are observed to suffer from error floors. LP
decoding, on the other hand, is guaranteed to correct some fraction of
errors in bit flipping channels (see e.g.~\cite{feldman2007lp,
  daskalakis2008probabilistic, arora2012message}) and, empirically,
has not been observed to suffer from an error floor 
for some practical long block length codes (cf.~\cite{barman2013decomposition}).

Despite having these advantages, LP decoding has much higher
computational complexity when implemented with generic LP
solvers. Some efficient LP decoding algorithms have been proposed,
such as~\cite{vontobel2006towards, taghavi2008adaptive,
  burshtein2009iterative, taghavi2011efficient, barman2013decomposition, zhang2013efficient, zhang2013large}. Among these
works,~\cite{ barman2013decomposition}
developed a distributed, scalable approach based on the alternating
direction method of multipliers
(ADMM). In~\cite{barman2013decomposition},
the authors exploit the geometry of the LP decoding constraints and
develop a decoding algorithm that has a message passing
interpretation. The simulation results
from~\cite{barman2013decomposition} reveal a drawback for LP decoding:
for all codes simulated in~\cite{barman2013decomposition}, LP decoding performs worse than BP
decoding at low SNRs in terms of word-error-rate (WER) performance. This suggests that an ideal decoder should
perform like BP at low SNRs and have no error floor effects at high
SNRs.

In this paper, we focus on improving LP decoding by changing the LP decoding objective. 
We first propose a class of decoders that we term the \emph{ADMM penalized decoders}. We then
derive another class of decoders based on the ADMM penalized decoder called the 
\emph{reweighted LP decoder}. 
All these new decoders are able to overcome the low SNR disadvantages
of LP decoding while keeping an important aspect of LP decoding: 
both decoder classes are not observed to suffer from an error floor. Furthermore, although not as efficient 
as the ADMM penalized decoder, the reweighted LP decoder
has a theoretical guarantee on error correction capability.

The key idea of the ADMM penalized decoder is to introduce a penalty term to the objective of LP
decoding. This penalty term penalizes non-integral
solutions by having the property of costing less at $0$ and $1$.
It is worth mentioning that the ADMM algorithm can be applied to generic constraint 
satisfaction problems~\cite{derbinsky2013improved}. Connecting to the generic framework
in~\cite{derbinsky2013improved} the ADMM penalized
decoder can be considered as having an objective 
that is the sum of channel evidences and a ``soft constraint'', i.e. the penalty term, enforcing variables to be in the set $\{0,1\}$, thereby
avoiding fractional solutions.
It is also well known that sum-product BP decoding is related to minimizing the
\textit{Bethe free energy} at temperature equals to one~\cite{vontobel2013counting}. The Bethe free energy is the sum of the LP
decoding objective and the \textit{Bethe entropy function}. Our
method can be considered as replacing the Bethe entropy function with a controllable penalty
term to retain the high SNR performance of LP decoding. 
The choice of this penalty is not unique. In this work,
we focus on using the $\ell_1$ norm and the $\ell_2$ norm as penalties. 
We then apply ADMM to this optimization problem and obtain a low complexity algorithm similar to~\cite{barman2013decomposition}. When equipped with efficient projection onto parity polytope sub-routine already 
developed in~\cite{barman2013decomposition} (and subsequently improved in~\cite{zhang2013efficient, zhang2013large}), 
we demonstrate that the decoding speed of the ADMM penalized decoder is much faster than our baseline 
sum-product BP decoder.

Although low SNR error performance can be verified through simulation, it is infeasible to do run simulations at high SNRs.
This means that we are uncertain about the high SNR behavior of the decoder, which becomes a crucial problem when applying this decoder to ultra-low error rate applications such as fiber-optic channels and data storage. To study the performance of the decoder at high SNRs, we adopt a numerical approach following the methodology in~\cite{stepanov2006instanton} where the concept of instanton is used. Instantons are the most probable noise vectors causing the decoder to fail and by searching for these noise vectors for a decoder/code pair, 
one can predict the high SNR behavior of that decoder/code pair (see also~\cite{chilappagari2009instanton}). 
We propose an instanton search algorithm for the ADMM penalized decoder and then apply it to the $[155,64]$ Tanner code~\cite{tanner2001class} and a $[1057, 813]$ LDPC code~\cite{Database}. We show that the instanton information we obtained provides good predictions for WER curve at high SNRs. In addition, our results suggest that the ADMM penalized decoder can suffer from trapping sets, which have been widely studied for message passing decoders (e.g.~\cite{mackayPostol:03}). This means that performance of the ADMM penalized decoder can be improved by code design techniques that remove trapping sets (e.g.~\cite{wang2013hierarchical}).

Unlike LP decoding, it is hard to prove guarantees for the ADMM penalized decoder
 because the latter is trying to solve a non-convex objective. 
Motivated by this problem, we develop a reweighted LP approximation of the penalized objective. 
By taking the linear term from the Taylor expansion of the $\ell_1$ penalized objective, we develop a new
reweighted LP decoder. Compared to the ADMM penalized decoder, the reweighted LP decoder has
worse empirical WER performance and higher computational complexity. However we are able to show 
that this decoder has an improved theoretical recovery threshold
compared to LP decoding for bit flipping channels. In addition, reweighted LP decoding still improves LP decoding significantly in terms of SNR performance. 

We note that some other authors have sought to improve the error performance
of LP decoding. One main focus in the literature has been tightening the polytope
under consideration (e.g.~\cite{tanatmis2010aseparation,dimakis2009guessing,yufit2011efficient, zhang2012adaptive}).
It is also possible to improve LP decoding by branch-and-bound type of algorithms such as those proposed in~\cite{yufit2011efficient, draper2007ml}. We refer interested readers to~\cite{helmling2012mathematical} for 
a more comprehensive review on methods to improve LP decoding. In an inspiring paper~\cite{khajehnejad2012reweighted}, Khajehnejad \emph{et al.} 
propose solving a second, reweighted LP, if LP decoding fails. This reweighted LP is constructed using the 
fractional output of the LP decoder, i.e., the pseudocodeword solution. Khajehnejad \emph{et al.} not only proved some useful theoretical tools for LP decoding but also showed via empirical results that their reweighted LP decoding improves LP decoding significantly.
Compared to the algorithm by Khajehnejad \emph{et al.} in~\cite{khajehnejad2012reweighted}, our reweighted LP decoding has two advantages: First, although the theorems we prove are for bit flipping channels, our reweighted LP decoding algorithm can also be applied to the AWGN channel. Second, our algorithm allows multiple rounds of reweighting in order to achieve better error performance, albeit without the guarantees of the first reweighted step. 

We summarize our major contributions in this paper:
\begin{itemize}
\item We introduce a new class of decoding objectives constructed by penalizing the LP decoding objective. (Section~\ref{subsec.generalpenalty})
\item We derive ADMM algorithms for decoders that either use the $\ell_1$ norm or the $\ell_2$ norm as a penalty function. (Section~\ref{subsection.penaltyfunctions}) 
\item We show that ADMM penalized decoding (i) improves LP decoding significantly in terms of both WER and decoding complexity and (ii) has high SNR behavior similar to LP decoding at least for the codes and SNRs we simulated. The proposed decoders achieve or outperform BP decoding at all SNRs in terms of both WER and decoding complexity. (Section~\ref{sec.numerical_penalty}) 
\item We develop an instanton search algorithm for the ADMM penalized decoder for the AWGN channel. (Section~\ref{subsection.instanton_search_algorithm})
\item  We show that instanton information can be used to predict the slope of the WER curve at high SNRs. Further, we observe that failures of the ADMM penalized decoder are related to trapping sets. (Section~\ref{sec.sim.instanton})
\item We develop a reweighted LP decoder that has a theoretical guarantee for bit flipping channels. (Section~\ref{sec.reweightedLP} and~\ref{subsec.rlp_thm}) 
\item We show by simulation results that our reweighted LP decoding outperforms LP decoding significantly in terms of SNR performance. (Section~\ref{sec.numerical.reweightedlp})
\end{itemize}

\section{Preliminaries}
\label{sec.LPprelim}
\subsection{LP decoding of LDPC codes}
Through out this paper, we consider LDPC codes of length
$\blocklength$ each specified by an $\checknumber\times\blocklength$
parity check matrix $\bfH$. Let $\cI:=\{1,2,\dots,\blocklength\}$
denote the index set of codeword symbols and
$\cJ:=\{1,2,\dots,\checknumber\}$ denote the index set of checks.  Let
$d_j,j\in\cJ$ be the degree of check node $j$. Let $\Ne_v(i)$ and 
$\Ne_c(j)$ be the sets of neighbors of variable node $i$ and check node $j$ respectively. 
We use $|\cdot|$ to denote the cardinality of a set and hence $d_j = |\Ne_c(j)|$.
Consider memoryless binary-input
symmetric-output channels, let $\inputsp = \{0,1\}$ be the
input space and $\outputsp$ the output space. Let $W:\inputsp \times
\outputsp \mapsto \mathbb{R}_{\geq0}$ be the channel law defined by
conditional probability $W(y|x)$, where $y\in \outputsp$,
$x\in\inputsp$. Since the channel is memoryless,
$\Pr_{\bfY|\bfX}(\bfy|\bfx) = \prod_{i = 1}^\blocklength W(y_i|x_i)$ for
$\bfx\in\inputsp^\blocklength, \bfy \in\outputsp^\blocklength $. Using these notations, ML decoding can be
stated as
  \begin{equation*}
 	 \begin{split}
	\min \quad \bfllr^T\bfx\quad
	\suchthat \quad  \bfx \in \codebook,
	\end{split}
  \end{equation*}
where $\bfgamma$ is the negative log-likelihood ratio, defined by
$\llr_i = \log\left(\frac{W(y_i|0)}{W(y_i|1)}\right)$ and $\codebook$
is the set of possible codewords. Feldman \textit{et al}. relax this
integer program to the following linear
program~\cite{feldman2005using}
  \begin{equation*}
 	 \begin{split}
	\min &\quad  \bfllr^T\bfx \quad\\
\suchthat &\quad  \bfx \in \fundpoly := \bigcap_{j\in\cJ}\conv(\codebook_j),
	\end{split}
  \end{equation*}
where $\conv(\cdot)$ is the operation of taking convex hull and
$\conv(\codebook_j)$ is the convex hull of codewords
 defined by the $j$-th row of the parity check
matrix $\bfH$. $\fundpoly$ is called the \textit{fundamental
  polytope}. An equivalent compact expression is
 \vspace{0.2cm}\\
\framebox[0.99\columnwidth]{
 \addtolength{\linewidth}{-4\fboxsep}%
 \addtolength{\linewidth}{-4\fboxrule}%
 \begin{minipage}{\linewidth}
  \begin{equation}
  \label{eq.lpfeldman}
 	 \begin{split}
	\mbox{\pLP:}\qquad\min &\quad  \bfllr^T\bfx \quad\\
\suchthat &\quad  \bfP_j \bfx \in \PP_{d_j},\forall j\in\cJ.
	\end{split}
  \end{equation}
 \end{minipage}
}\vspace{0.2cm} 

Here $\bfP_j$ is the operation of selecting the
sub-vector of $\bfx$ that participates in the $j$-th check. $\PP_{d_j}$ is
the parity polytope of dimension $d_j$. It is define as the convex hull
of all even-parity binary vectors of length $d_j$; in other words, the
convex hull of all codewords of the length-$d_j$ single parity check
code.

In the BSC, error performance of LP decoding can be improved when
a second (reweighted) LP is used whenever \pLP fails~\cite{khajehnejad2012reweighted}.
We denote this reweighted LP as \pRLPKDHVB(Reweighted LP decoding by Khajehnejad-Dimakis-Hassibi-Vigoda-Bradley).
\vspace{0.2cm}\\
\framebox[0.99\columnwidth]{
 \addtolength{\linewidth}{-4\fboxsep}%
 \addtolength{\linewidth}{-4\fboxrule}%
 \begin{minipage}{\linewidth}
\begin{equation}
\label{eq.rlp_kdhb}
\begin{split}
\pRLPKDHVB :& \\
\min & \quad c_1\Vert (\bfx - \bfy)_\mcT\Vert_1 + c_2 \Vert (\bfx - \bfy)_{\mcT^c}\Vert_1\\
 \suchthat &\quad  \bfP_j \bfx \in \PP_{d_j},\forall j\in\cJ.
 \end{split}
\end{equation}
where $c_1 < 0$ and $c_2>0$ are constants, $\mcT \subset \{1,2,\dots,\blocklength\}$ is a fixed subset and $\mcT^c$ is the compliment of $\mcT$.\footnote{The set $\mcT$ is chosen using the both the pseudocodeword and the channel output. See details in~\cite{khajehnejad2012reweighted}.}
\end{minipage}
}\vspace{0.2cm}

Note that this is still a linear program since for a given fixed $\bfy \in \{0,1\}^{\blocklength}$, the objective is affine linear with respect to $\bfx$. 

\subsection{ADMM formulation for general optimization decoding problem}
We now turn to review an efficient algorithm for solving the problems mentioned
above, namely, the ADMM formulation presented in~\cite{barman2013decomposition}. 
We generalize this formulation to 
optimization problems over the fundamental polytope $\fundpoly$, which then 
enables us to reformulate \eqref{eq.rlp_kdhb} easily. We
refer the readers to~\cite{boyd2010distributed} for a comprehensive
description of ADMM. When using ADMM, a crucial subroutine is the projection onto $\PP_{d_j}$.
Barman \textit{et al.} characterized the geometry of $\PP_{d_j}$ 
in~\cite{barman2013decomposition} and develop
an efficient algorithm for projecting onto this object. Some recent works
further improved the projection algorithm. Zhang and Siegel proposed an algorithm
based on cuts in~\cite{zhang2013efficient}. Zhang \textit{et al.} further
leveraged the ``two-slice'' property developed in~\cite{barman2013decomposition}
and presented a linear time algorithm in~\cite{zhang2013large}.

To cast the LP decoding problem into the
framework of ADMM, the key is to introduce \textit{replica} vectors
$\bfz_j$, for all $j\in\cJ$ and write \eqref{eq.lpfeldman} in the following,
equivalent, form:
\begin{equation*}
\begin{split}
\min\quad & \bfgamma^T\bfx \\
\suchthat \quad & \forall j\in\cJ, \bfP_j \bfx = \bfz_j,\\
& \bfz_j \in \PP_{d_j}.
\end{split}
\end{equation*}
We now generalize the objective function above as
\begin{equation*}
\begin{split}
\min\quad & f(\bfx) \\
\suchthat\quad & \forall j\in\cJ, \bfP_j \bfx = \bfz_j,\\
& \bfz_j \in \PP_{d_j}.
\end{split}
\end{equation*}
Note that $f(\cdot)$ is not necessarily convex. The augmented Lagrangian~\cite{boyd2010distributed} for the above
problem is
\begin{equation*}
\mcL_{\penpara}(\bfx,\bfz,\lmbpara) = f(\bfx) + \sum_j \lmbpara_j(\bfP_j\bfx-\bfz_j) + \frac{\penpara}{2}\sum_j\Vert \bfP_j\bfx-\bfz_j\Vert_2^2,
\end{equation*}
where $\penpara$ is some constant. The ADMM update rules are
\begin{align}
\bfx\text{-update: }\bfx^{k+1} &= \argmin_{\bfx} \mcL_{\penpara}(\bfx,\bfz^{k},\lmbpara^{k}), \label{eq.admm_x_update}\\
\bfz\text{-update: }\bfz^{k+1} &= \argmin_{\bfz} \mcL_{\penpara}(\bfx^{k+1},\bfz,\lmbpara^{k}), \label{eq.admm_z_update}\\
\lmbpara\text{-update: }\lmbpara^{k+1} &=\lmbpara_j^k + \penpara\left(\bfP_j\bfx^{k+1}-\bfz_j^{k+1}\right). \label{eq.admm_lambda_update}
\end{align}

Note that \eqref{eq.admm_z_update} is minimizing over $\bfz_j$ and
therefore $f(\bfx)$ can be considered as a constant that does not
affect the solution to the problem. Hence the $\bfz$-update is identical to
that in~\cite{barman2013decomposition}. 
\eqref{eq.admm_lambda_update} is also not related to $f(\cdot)$. Therefore when dealing with
arbitrary objective function $f(\cdot)$, we need only modify the
$\bfx$-update,
vis-\`{a}-vis~\cite{barman2013decomposition}.
At this point, we should be very careful that $\mcL_{\penpara}(\bfx,\bfz,\lmbpara)$ should have a well
defined minimizer with respect to $\bfx$, provided that $\bfz$ and 
$\bflambda$ are fixed. This will allow us to determine $\bfx$-update uniquely.

Rephrasing \eqref{eq.admm_x_update} as a coordinate-wise update rule
for $\bfx$, we get:
\begin{equation}
\label{eq.admm_x_elementwise_update}
x_i = \Proj_{[0,1]}\left[\frac{1}{\degi}\left(\sum_j \left(\bfz_j^{(i)} - \frac{\lmbpara_j^{(i)}}{\penpara}\right) - \frac{1}{\penpara}\left(\nabla_{\bfx} f\right)_i\right)\right],
\end{equation}
where the superscript ``$i$'' denotes the entries in $\bfz_j,\bflambda_j$ that are associated with variable $i$ (cf.~\cite{barman2013decomposition}). $\Proj_{[0,1]}$ is the operation of projecting the results onto $[0,1]$.
For any objective function $f(\cdot)$, we can derive $\nabla_{\bfx} f$ to determine the update rule. 
\begin{example}
(Reweighted LP decoding in~\cite{khajehnejad2012reweighted})
Let $f(\bfx)$ be the objective in~\eqref{eq.rlp_kdhb}. Since $x_i \in [0,1]$, if $y_i = 0$ then $|x_i - y_i| = x_i$. If $y_i = 1$, then $|x_i - y_i| = - x_i$. Therefore $\left(\nabla_{\bfx}f\right)_i = c_v$ if $y_i = 0$ where $v = 1$ if $i\in \mcT$ and $v = 2$ if $i\in \mcT^c$. Similarly, $\left(\nabla_{\bfx}f\right)_i = -c_v$ if $y_i = 1$. Combining the two cases, we obtain the $\bfx$-update rule in~\eqref{eq.amin_reweightedlp_x_update}.
\begin{equation}
\label{eq.amin_reweightedlp_x_update}
x_i = 
\begin{cases} \Proj_{[0,1]}\left[\frac{1}{\degi}\left(\sum_j \left(\bfz_j^{(i)} - \frac{\lmbpara_j^{(i)}}{\penpara}\right) - \frac{c_v}{\penpara}\right)\right] & \text{if $y_i = 0$,}
\\
\Proj_{[0,1]}\left[\frac{1}{\degi}\left(\sum_j \left(\bfz_j^{(i)} - \frac{\lmbpara_j^{(i)}}{\penpara}\right) + \frac{c_v}{\penpara}\right)\right] &\text{if $y_i = 1$.}
\end{cases}
\end{equation}
\end{example}

Since ADMM is not restricted to linear programs, we can use ADMM algorithm 
for any objective function on $\fundpoly$. 
This means that we have the freedom to choose a good decoding objective. 
We leverage this observation in the next section and develop decoding
methods that incorporate a penalty term into the original linear objective. 
\section{Using penalty functions for decoding}
\label{sec.penaltyfunction}
LP decoding fails
whenever it outputs a fractional solution. This solution is called a
\textit{pseudocodeword}. Since pseudocodewords are directly
related to the fractional vertices of the fundamental polytope
$\fundpoly$, it is not surprising that error performance can be
improved by tightening the relaxation of $\fundpoly$. ``Cutting
plane'' and ``branch-and-bound'' are two types of techniques that are
used to improve error performance of LP
decoding~\cite{helmling2012mathematical}.

In this work, we take a different angle to improve LP decoding. We
are partially inspired by~\cite{khajehnejad2012reweighted} because the
authors improve LP decoding not by tightening the constraints, but by
properly changing the objective. Following this observation, we
propose a method that penalizes pseudocodewords using functions that
 make integral solutions more favorable than
fractional solutions. 

Our approach is also similar to the well known LASSO
algorithm~\cite{tibshirani1996regression}: the penalty in our
approach can be considered as a ``regularizer''. However there is one
critical difference: the objective in our approach is non-convex. We
show in Section~\ref{sec.reweightedLP} that we can approximate an $\ell_1$
penalized objective using reweighted LPs which have theoretical
guarantees under certain regimes. 

\subsection{Family of penalty functions for decoding}
\label{subsec.generalpenalty}
We now formally state the ADMM penalized decoding problem. We propose applying the ADMM algorithm to the following optimization problem:
\vspace{0.2cm}\\
\framebox[\columnwidth]{
 \addtolength{\linewidth}{-4\fboxsep}%
 \addtolength{\linewidth}{-4\fboxrule}%
 \begin{minipage}{\linewidth}
  \begin{equation}
\label{eq.generalpenalty}
\begin{split}
\mbox{\pPD:} \qquad \min \quad & \bfgamma^T \bfx  + \sum_i g(x_i) \\
\suchthat\quad  & \bfP_j \bfx \in \PP_{d_j}
\end{split}
  \end{equation}
  where $g:[0,1]\mapsto \mathbb{R}\bigcup\{\pm \infty\}$ is a \textit{penalty function}. 
 \end{minipage}
}\vspace{0.2cm}

This decoder is denoted as \pPD (ADMM penalized decoding). For reasons to be discussed, we focus on penalty functions with the following properties:
\vspace{0.2cm}\\
\fbox{
 \addtolength{\linewidth}{-4\fboxsep}%
 \addtolength{\linewidth}{-4\fboxrule}%
 \begin{minipage}{\linewidth}
 \pPD penalty properties:
\begin{enumerate}
\item[(i)] $g$ is an increasing function on $(0,1/2)$;
\item[(ii)] $g(x) = g(1- x)$ for $x\in[0,1]$;
\item[(iii)] $g$ is differentiable on $(0,1/2)$.
\item[(iv)] $g$ is such that the $\bfx$-update of ADMM is well defined.
\end{enumerate}
\end{minipage}
}\vspace{0.2cm}

Property (i) penalizes fractional solutions. Property (ii) ensures
that the decoding error is independent with respect to the transmitted
codeword (see Theorem~\ref{thm.cw_independent}). We impose property
(iii) and (iv) so that we can determine the ADMM updates
with explicit form~(cf. $\nabla_{\bfx}f$
in~\eqref{eq.admm_x_elementwise_update}).
\begin{example}
\label{example.penaltyfunctions}
We present several examples satisfying these properties. We further discuss these penalties in Section~\ref{subsection.penaltyfunctions}: $g_1(x) = -\alpha_1|x-0.5|$,
$g_2(x) = -\alpha_2(x-0.5)^2$ and $g_3(x) = -\alpha_3\log(|x-0.5|)$, where $\alpha_i,i=1,2,3$ are designer-chosen constants. We plot these functions in Fig.~\ref{fig.penaltyfunction}. It is easy 
to observe the desired properties in these functions.
\begin{figure}
\psfrag{&A}{\scriptsize{penalty, $g(x)$}}
\psfrag{&B}{\scriptsize{$x$}}
\psfrag{&C}{\scriptsize{$-\log(|x-0.5|)$}}
\psfrag{&D}{\scriptsize{$-(x-0.5)^2$}}
\psfrag{&E}{\scriptsize{$-|x-0.5|$}}
\psfrag{0}{\scriptsize{$0$}}
\psfrag{0.5}{\scriptsize{$0.5$}}
\psfrag{1}{\scriptsize{$1$}}
    \begin{center}
    \includegraphics[width=3.5in]{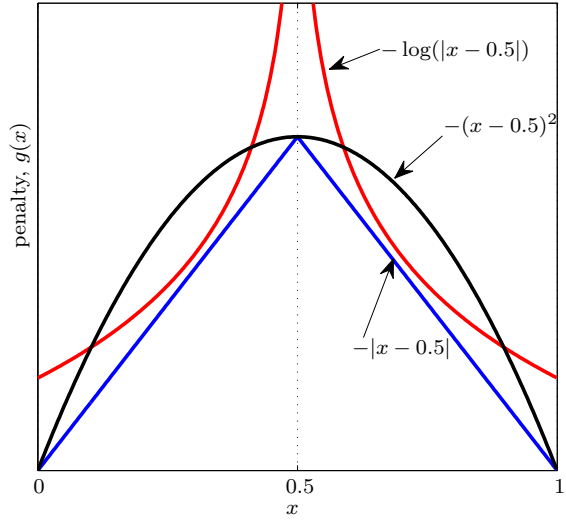}
    \end{center}
    \caption{Three penalty functions. Each functions is shifted and scaled for better view.}
    \label{fig.penaltyfunction}
\end{figure}
\end{example}
 
Note that introducing the
penalty term turns the LP into a nonlinear and non-convex
program. Thus one should not expect to find the global
solution. However we can still perform ADMM updates to obtain a local
minimal. More discussions on solving non-convex problems using ADMM can be
found in \cite[Sec.~9]{boyd2010distributed}.

In the ADMM $\bfx$-updates we solve for local solutions by solving for the stationary points. Note that $g$ is an increasing function on $(0,1/2)$. Thus, if there are multiple stationary points, we always choose the one that is farthest from $0.5$. The ADMM update steps are summarized in Algorithm~\ref{Algorithm:ADMM}.

\begin{algorithm}
\caption{\pPD. \textbf{Input}: Received vector $\bfy \in \Sigma^N$.
\textbf{Output}: Decoded vector $\bfx$.}
\label{Algorithm:ADMM}
\begin{algorithmic}[1]

\STATE Construct the $d_j \times N $ selection matrix $\bfP_j$ for all $j \in
\mathcal{J}$ based on the parity check matrix $\bfH$.

\STATE Construct the log-likelihood ratio $\bfgamma$ based on received vector $\bfy$. 

\STATE Initialize $\bflambda_j$ as the all zeros vector and $\bfz_j$ as the all $0.5$ vector for
all $j \in \mathcal{J}$. Initialize iterate $k = 0$. For simple notation, we drop iterate $k$ except for determining stopping criteria.

\REPEAT

\FORALL{ $ i \in \mathcal{I} $}
\STATE \label{algstep.xupdate1} Update $x_i$ by choosing the farthest root from 0.5 for equation 

$x_i \leftarrow \frac{1}{\degi}\left(\sum_j \left(\bfz_j^{(i)} - \frac{\lmbpara_j^{(i)}}{\penpara}\right) - \frac{1}{\penpara}(\gamma_i + g^\prime(x))\right)$. 
\STATE \label{algstep.xupdate2} $x_i \leftarrow \Proj_{[0,1]}(x_i)$ (project $x_i$ onto $[0,1]$).
\ENDFOR
\FORALL{ $ j \in \mathcal{J} $}

\STATE Set $ \bfv_j  \leftarrow \bfP_j \bfx + \bflambda_j /\mu $.

\STATE Update $ \bfz_j \leftarrow \Proj_{\PP_{d_j}} (\bfv_j) $ where
$\Proj_{\PP_{d_j}} (\cdot)$ means project onto the parity polytope.

\STATE Update $\bflambda_j \leftarrow \bflambda_j + \mu \left( \bfP_j \bfx - \bfz_j\right) $.
\ENDFOR
\STATE $ k \! \leftarrow \! k+1$.
\UNTIL{ $ \sum_j { \| \bfP_j \bfx^k - \bfz^k_j \|^2_{2} } < \epsilon^2 Md_j $ \\and  $\sum_j { \| \bfz^{k}_j - \bfz^{k-1}_j \|^2_{2} } < \epsilon^2 Md_j$}\\
{\bf return} $\bfx$.
\end{algorithmic}
\end{algorithm}

\begin{proposition}
If the output from Algorithm~\ref{Algorithm:ADMM} is an integral solution, then it is an valid codeword. 
\end{proposition}
\begin{proof}
This comes directly from the fact that the feasible integral solutions in the fundamental polytope are valid codewords.
\end{proof}

In the next proposition, we show that the penalty term does not prefer a non-ML codeword over the ML codeword. Namely, the ML codeword is the most competitive integral solution. This indicates that our choice of penalty function is reasonable.
\begin{proposition}(Ideal ML test)
\label{prop.idealmltest}
If the output from Algorithm~\ref{Algorithm:ADMM} is the global minimizer of problem~\eqref{eq.generalpenalty} and it is also an integral solution, then it is the ML solution.
\end{proposition}
\begin{proof}
Prove by contradiction: Let $f(\bfx) =  \bfgamma^T\bfx  + \sum_i g(x_i)$. Suppose the decoder outputs an integral solution $\bfx^*$ that is not equal to the ML solution $\tilde{\bfx}$. By property (i) and (ii) of $g$,
\begin{equation*}
g(0) = g(1) < g(x),\qquad \forall x\in(0,1)
\end{equation*}
Then 
\begin{align*}
f(\bfx^*) = \bfgamma^T\bfx^*  + \sum_i g(x_i^*) &= \bfgamma^T\bfx^* + ng(0),\\
f(\tilde{\bfx}) = \bfgamma^T\tilde{\bfx}  + \sum_i g(\tilde{x}_i) &= \bfgamma^T\tilde{\bfx} + ng(0).
\end{align*}
Since $\tilde{\bfx}$ is the ML solution, $\bfgamma^T\tilde{\bfx}<\bfgamma^T\bfx^*$, therefore $f(\tilde{\bfx}) < f(\bfx^*)$, contradicting with the assumption that $\bfx^*$ is the global minimizer.
\end{proof}

In general, we cannot determine whether or not a feasible point is the global minimizer of a
non-convex optimization problem. But even if we cannot tell whether or not
the output from the decoder is the global minimizer, we can use the ML
certificate property of LP decoding to test whether or not the solution
obtained is the ML solution. Namely, if we find an integral solution
we can run ADMM on the original LP problem~\eqref{eq.lpfeldman}, using
this integral solution as the initial point, thereby testing whether or
not this integral solution is the ML solution. Formally,
\begin{proposition}(Weak ML test)
If the output from Algorithm~\ref{Algorithm:ADMM} is an integral solution, we do one more iteration of the algorithm without the penalty function. If the new solution obtained is the same as the previous one then this integral solution is the ML solution.
\end{proposition}
\begin{proof}
By dropping the penalty function, we obtain the original LP decoding problem. By the ML certificate of LP decoding, the solution is the ML solution.
\end{proof}

\begin{theorem}
\label{thm.cw_independent}
(All-zero assumption) If the channel is symmetric then the probability that Algorithm~\ref{Algorithm:ADMM} fails
is independent of the codeword that was transmitted.
\end{theorem}
\begin{proof}
See Appendix~\ref{appendix.cwindep_proof}.
\end{proof}
In general, problem~(\ref{eq.lpl1penalty}) is non-convex
and therefore the output from the decoder should be different if
it is initialized at different points. We show the all-zero assumption
property under two different initializations in
Corollary~\ref{crly.initialpoints}.
\begin{corollary}
\label{crly.initialpoints}
The probability of error is independent of the codeword that was transmitted if Algorithm~\ref{Algorithm:ADMM} is initialized in the following ways:
\begin{enumerate}
\item For all $j\in\cJ $, $\bflambda_j = 0$ and each entry in $\bfz_j$ is i.i.d. following uniform distribution on $[0,1]$.
\item For all $j\in\cJ $, $\bflambda_j = 0$ and $\bfz_j$ is obtained from the solution of LP
decoding~\eqref{eq.lpfeldman} (i.e., at the pseudocodeword solution for \pLP).
\end{enumerate}
\end{corollary}
\begin{proof}
See Appendix~\ref{appendix.cwindep_proof_corollary}.
\end{proof}
 Corollary~\ref{crly.initialpoints} implies that if we construct a two-stage decoder by concatenating an LP decoder followed by \pPD, then the probability of error for this two-stage decoder is independent of the codeword that was transmitted.

\subsection{Penalty functions using the $\ell_1$ and $\ell_2$ norms}
\label{subsection.penaltyfunctions}
We have already presented some examples of penalty functions in Example~\ref{example.penaltyfunctions}. We first discuss the following penalty functions in details:
\begin{align*}
g_1(x) &= -\alpha_1|x-0.5|,\\
g_2(x) &= -\alpha_2(x-0.5)^2.
\end{align*}
where $\alpha_1, \alpha_2 > 0$ are constants and are called the \emph{penalty coefficients}. We then briefly discuss the $\log$ function and other possibilities.

The respective decoding algorithms for these two functions
are identical to Algorithm~\ref{Algorithm:ADMM} except that
Steps~6 and~7 are replaced by
specific update rules
(cf.~\eqref{eq.alg_ell1penalty} and~\eqref{eq.alg_ell2penalty} 
respectively). Note that if $\alpha_i =
0$  ($i \in\{1,2\}$), the problems will become LP
decoding~\eqref{eq.lpfeldman}. We show that $g_1$ could yield
2 stationary points during $\bfx$-update in ADMM. We choose the
stationary point farthest from $0.5$ using simple thresholding.

\subsubsection{$\ell_1$ penalty function}
The decoding algorithm using $g_1(x) = -\alpha_1|x-0.5|$ can be considered as running ADMM on the following problem:
\begin{equation}
\label{eq.lpl1penalty}
\begin{split}
\min \quad & \bfgamma^T \bfx  - \alpha_1 \Vert \bfx-0.5\Vert_1 \\
\suchthat\quad  & \bfP_j \bfx \in \PP_{d_j}
\end{split}.
\end{equation}
We apply~\eqref{eq.lpl1penalty} to Step~\ref{algstep.xupdate1} in Algorithm~\ref{Algorithm:ADMM}: The objective function becomes
$f(\bfx) = \bfgamma^T\bfx  - \alpha_1\Vert \bfx-0.5\Vert_1$. Therefore $\left(\nabla_{\bfx}f\right)_i = \bfgamma_i - \alpha_1\sgn(x_i - 0.5)$. Letting 
\begin{equation}
\label{eq.t_i}
t_i := \sum_j \left(\bfz_j^{(i)} - \frac{\lmbpara_j^{(i)}}{\penpara}\right) - \bfgamma_i/\mu,
\end{equation}
we get the following equation:
\begin{equation}
\label{eq.ellone_x_update}
x_i = \frac{1}{\degi}\left(t_i + \frac{\alpha_1}{\mu}\sgn(x_i - 0.5)\right).
\end{equation}
Now consider the relationship between $x_i$ and $t_i$, which is shown in Fig.~\ref{fig.updaterel}. The solid lines in the figure represent legitimate solution pairs $(x_i,t_i)$ for~\eqref{eq.ellone_x_update}. We can see from the Fig.~\ref{fig.updaterel} that there are cases where for a given $t_i$, there are two solutions for $x_i$. We update $x_i$ with the one that is farthest from $0.5$. So the threshold is $t_i = \frac{\degi}{2}$. The $\bfx$-update rule is summarized in~\eqref{eq.alg_ell1penalty}:
\begin{equation}
\label{eq.alg_ell1penalty}
x_i= 
\begin{cases} 
\Proj_{[0,1]}\left( \frac{1}{\degi}  \left( t_i + \frac{\alpha}{\mu}\right) \right) & \text{if $t_i \geq \degi/2$,}\\
\Proj_{[0,1]}\left( \frac{1}{\degi}  \left( t_i - \frac{\alpha}{\mu}\right) \right) & \text{if $t_i < \degi/2$.}
\end{cases}
\end{equation}

\begin{figure}[htbp]
\psfrag{&A}{{$x_i$}}
\psfrag{&B}{{$t_i$}}
\psfrag{&C}{{$\frac{1}{2}$}}
\psfrag{&D}{{$\frac{\degi}{2}$}}
\psfrag{&E}{{slope:$\frac{1}{\degi}$}}
    \begin{center}
    \includegraphics[width=3in]{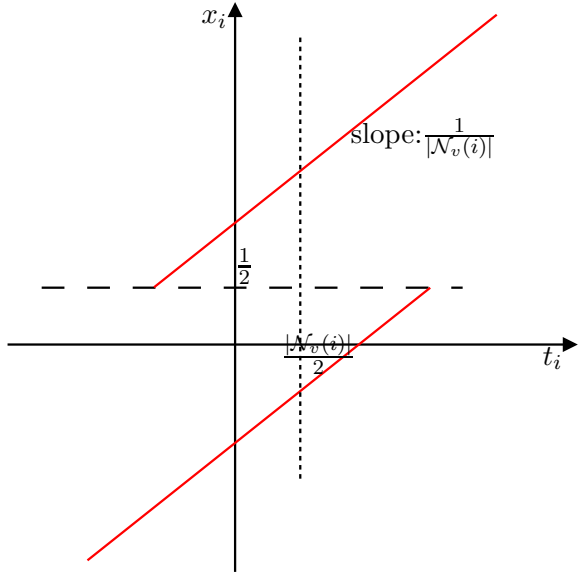}
    \end{center}
    \caption{Relationship between $x_i$ and $t_i$ in~\eqref{eq.ellone_x_update}}
    \label{fig.updaterel}
\end{figure}

\subsubsection{$\ell_2$ penalty function}
The decoding algorithm that uses $g_2(x) = -\alpha_2(x-0.5)^2$ runs ADMM on the following problem:
\begin{equation}
\label{eq.lpl2penalty}
\begin{split}
\min \quad & \bfgamma^T \bfx  - \alpha_2 \Vert \bfx-0.5\Vert_2^2 \\
\suchthat\quad  & \bfP_j \bfx \in \PP_{d_j} 
\end{split}.
\end{equation}
Note that in order to make the $\bfx$-update well defined, $\mcL_{\penpara}(\bfx,\bfz,\lmbpara)$ should be convex with respect to $\bfx$. Therefore we need to set $\alpha_2 \leq \frac{\degi \mu}{2}$ for all $i \in \cI$. We write the ADMM $\bfx$-update as:
$\left(\nabla_{\bfx}f\right)_i = \bfgamma_i - 2\alpha_2(x_i - 0.5)$. In this case, there is only one possible stationary point. The update rule is summarized in~\eqref{eq.alg_ell2penalty}:
\begin{equation}
\label{eq.alg_ell2penalty}
x_i =\Proj_{[0,1]}\left( \frac{1}{\degi - \frac{2\alpha_2}{\mu}} \left( t_i - \frac{\alpha_2}{\mu}\right) \right),
\end{equation}
where $t_i$ is defined in~\eqref{eq.t_i}.

\subsection{Other possible penalty functions}
We briefly discuss three other penalty functions. 
\begin{align*}
g_3(x) &= -\alpha_3\log(|x-0.5|),\\
g_4(x) &= 
\begin{cases} 0 & \text{if $x = 0$ or $1$,}\\
+\infty &\text{otherwise.}
\end{cases} \\
g_5(x) &=\alpha_5 H_B(x),
\end{align*}
where $H_B(\cdot)$ is the binary entropy function. 

The ADMM $\bfx$-update rule for $g_3(x)$ is similar to that for $g_1(x)$ where there are two stationary points and one needs to choose the one farthest from $0.5$. We observe that the empirical performance of $g_3(x)$ is  worse than $g_1(x)$ and $g_2(x)$. $g_4(x)$ forbids all non-integral solutions of $\bfx$ and yields projection on $\{0,1\}$ in $\bfx$-update. We simulated $g_4$ and ADMM does not converge. For $g_5(x)$, there are two critical problems. First, the ADMM $\bfx$-update for this penalty does not have an explicit expression. One can use the Newton's method (cf. e.g.~\cite{ypma1995historical}) to solve for a stationary point but this increases the complexity for the $\bfx$-update. Second, the derivative of binary entropy function goes to infinity at $x= 0$ and $x=1$, this creates a numerical problem for implementation. 

We verified via simulations that these functions are not better than $g_1(x) = -\alpha_1 |x- 0.5|$ and $g_2(x) = -\alpha_2 (x- 0.5)^2$. Thus we do not pursue these methods any further. Interested readers are referred to our previous work~\cite{liu2012suppressing} for simulation results on the decoder using the $\log$ penalty. (We also would like to encourage interested readers to explore other possibilities.)
\subsection{Numerical results of \pPD}
\label{sec.numerical_penalty}
We now examine the performance of \pPD empirically. We will study the decoder in depth via an instanton analysis in Section~\ref{sec.instanton}. 
In this section, we present numerical results for \pPD. 
In Section~\ref{sec.sim.wer}, we present WER comparisons
among \pPD with both $\ell_1$ and $\ell_2$ penalties and other decoders (BP and LP). 
In Section~\ref{sec.sim.parameter}, we study the effect of
different parameters. We show that even though the parameter search
space is large, it is easy to obtain a high performing parameter set.

\subsubsection{WER performances}
\label{sec.sim.wer}
We simulate the $[2640,1320]$ ``Margulis''
code~\cite{ryanlin2009channel} for the AWGN channel. We use the following parameter 
settings for \pPD: $\mu = 3$, $\epsilon = 10^{-5}$ and maximum number of iterations 
$T_{\max} = 1000$. For the $\ell_1$ and $\ell_2$ penalties, we pick $\alpha_1 = 0.6$
and $\alpha_2 = 0.8$ respectively (cf.~\eqref{eq.lpl1penalty}
and \eqref{eq.lpl2penalty}). We defer discussions on parameter choice to
Section~\ref{sec.sim.parameter} and Section~\ref{sec.sim.instanton}. 
\begin{figure}[h]
\psfrag{&ADMMLP}{\scalebox{.6}{ADMM LP decoding} }
\psfrag{&BPSaturation}{\scalebox{.6}{BP decoding (Ryan and Lin)}}
\psfrag{&BPNonSat}{\scalebox{.6}{Non-saturating BP decoding}}
\psfrag{&ADMM-L1PenaliedDecoderpa}{\scalebox{.6}{\pPD, $\ell_1$ penalty }}
\psfrag{&ADMM-L2PenaliedDecoder}{\scalebox{.6}{\pPD, $\ell_2$ penalty }}
\psfrag{&SNR}{\scalebox{.6}{\hspace{-1mm}$E_b/N_0(\dB)$}}
\psfrag{&WER}{\scalebox{.6}{WER}}

\psfrag{1}{\scriptsize{$1$}}
\psfrag{1.5}{\scriptsize{$1.5$}}
\psfrag{2}{\scriptsize{$2$}}
\psfrag{2.5}{\scriptsize{$2.5$}}
\psfrag{3}{\scriptsize{$3$}}

\psfrag{10}{\scriptsize{\hspace{-1mm}$10$}}
\psfrag{-2}{\tiny{\hspace{-1mm}$-2$}}
\psfrag{-4}{\tiny{\hspace{-1mm}$-4$}}
\psfrag{-6}{\tiny{\hspace{-1mm}$-6$}}
\psfrag{-8}{\tiny{\hspace{-1mm}$-8$}}

\psfrag{e0}{\tiny{\hspace{-1mm}$0$}}

    \begin{center}
    \includegraphics[width=3.5in]{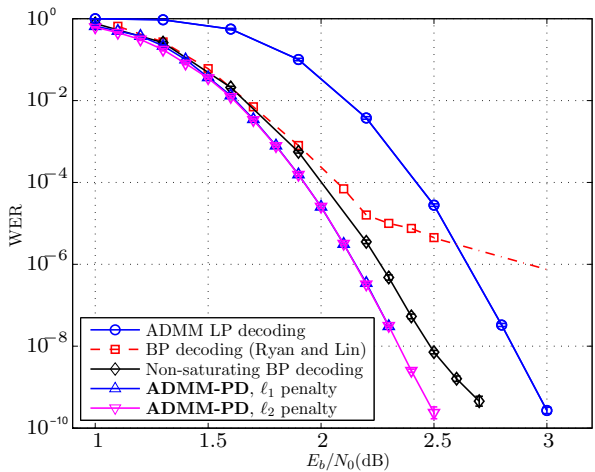}
    \end{center}
    \caption{Word-error-rate (WER) of the
      $[2640,1320]$ ``Margulis'' code plotted as a function of $E_b / N_0$ for the AWGN channel.}
    \label{fig.sim_margulis}
\end{figure}
In Fig.~\ref{fig.sim_margulis}, we
plot the WER as a function of SNR. For each data point, we also plot
the error bar. Data for BP and LP decoding are drawn from~\cite{barman2013decomposition}.
For $E_b/N_0 = 2.4\dB, 2.5\dB$, we collect respectively $121$ and $9$ word errors for 
\pPD with an $\ell_2$ penalty. Also note that we do not simulate \pPD with an $\ell_1$ penalty at $E_b/N_0 = 2.4\dB$ and $2.5\dB$ due to limited computational resources. All other data points (excluding those for BP and LP decoding) are based on more than $200$ word errors.

The ``non-saturating BP'' is a robust BP implementation that avoids saturating numerical issue of BP decoding (see~\cite{Butler2012Error} for details) and thus improves the decoding performance at high SNRs. We note that LP decoding shows worse performance than BP decoding at low SNRs, however BP decoding has an error floor at high SNRs. Even with non-saturating BP implementation, the slope of the WER curve is approximately constant at $E_b/N_0 = 2.3,2.4,2.5\dB$ and eventually decreases at $E_b/N_0 = 2.6,2.7\dB$. However the slope of the WER for LP decoding keeps steepening. The WER of \pPD combines the good characteristics of both LP and
BP decoders: \pPD using $\ell_1$ and $\ell_2$ penalties both perform as well as BP decoder at low SNRs. In addition, similar to LP decoder,
no error floor is observed for WERs above $10^{-10}$.

When comparing the $\ell_1$ penalty and the $\ell_2$ penalty, we do not observe significant 
differences. The $\ell_2$ penalty achieves slightly better error rate at $E_b/N_0 = 1.2,1.3$ and $1.4\dB$ in term of WER however both decoders achieve the same WER at other SNRs. We show in the next example that the $\ell_2$
penalty outperforms the $\ell_1$ penalty in a longer code.

It is worth mentioning that the slope of the WER curve for \pPD with an $\ell_2$ penalty does not increase at $E_b/N_0= 2.5\dB$. This high SNR behavior is further 
studied and discussed in Section~\ref{sec.sim.instanton}, where we 
show that the optimal parameter $\alpha_2$ in terms of WER performance should be small for high SNRs. 
The same argument also applies to the $\ell_1$ penalty. Due to 
limited computational resources, we did not simulate other values of $\alpha_2$ at high SNRs. 

\begin{figure}[h]
\psfrag{&ADMMLP}{\scalebox{.6}{ADMM LP decoding}}
\psfrag{&BPNonSat}{\scalebox{.6}{Non-saturating BP decoding}}
\psfrag{&ADMM-L1PenaliedDecoderP}{\scalebox{.6}{\pPD, $\ell_1$ penalty }}
\psfrag{&ADMM-L2PenaliedDecoder}{\scalebox{.6}{\pPD, $\ell_2$ penalty }}
\psfrag{&SNR}{\scalebox{.6}{\hspace{-1mm}$E_b/N_0(\dB)$}}
\psfrag{&WER}{\scalebox{.6}{WER}}

\psfrag{1}{\scriptsize{$1$}}
\psfrag{1.2}{\scriptsize{$1.2$}}
\psfrag{1.4}{\scriptsize{$1.4$}}
\psfrag{1.6}{\scriptsize{$1.6$}}
\psfrag{1.8}{\scriptsize{$1.8$}}
\psfrag{2}{\scriptsize{$2$}}
\psfrag{2.2}{\scriptsize{$2.2$}}
\psfrag{2.4}{\scriptsize{$2.4$}}
\psfrag{2.6}{\scriptsize{$2.6$}}

\psfrag{10}{\scriptsize{\hspace{-1mm}$10$}}
\psfrag{-2}{\tiny{\hspace{-1mm}$-2$}}
\psfrag{-4}{\tiny{\hspace{-1mm}$-4$}}
\psfrag{-6}{\tiny{\hspace{-1mm}$-6$}}
\psfrag{-8}{\tiny{\hspace{-1mm}$-8$}}

\psfrag{e0}{\tiny{\hspace{-1mm}$0$}}

    \begin{center}
    \includegraphics[width=3.5in]{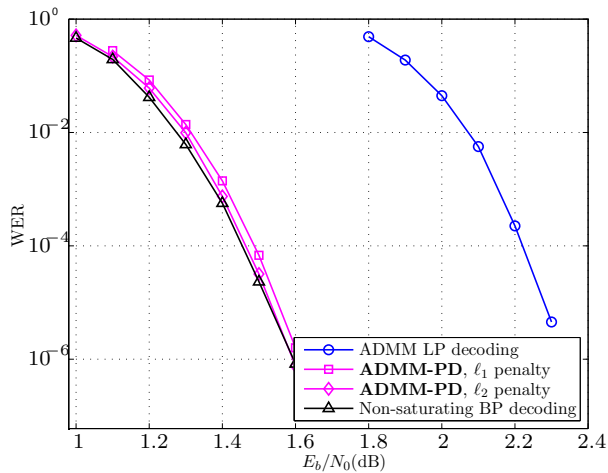}
    \end{center}
    \caption{Word-error-rate (WER) of the
      Mackay $[13298,3296]$ code plotted as a function of $E_b / N_0$ for the AWGN channel.}
    \label{fig.sim_mackay}
\end{figure}

In Fig.~\ref{fig.sim_mackay}, we plot the WER performance of the Mackay $[13298,3296]$ LDPC code obtained 
from~\cite{Database}. All data points are based on more than $200$ decoding errors.
Compared to the $[2640, 1320]$ Margulis code, this code is significantly longer. 
 The decoder's parameter settings are the same as the previous 
experiment. We observe from Fig.~\ref{fig.sim_mackay}
that there is a $0.8\dB$ SNR gap between LP decoding and BP decoding.
\pPD closes the SNR gap between BP and LP decoding. Further,
\pPD retains the good high SNR behaviors of LP decoding. The WER slope of \pPD
at $E_b/N_0 = 1.5$, $1.6\dB$ is steeper than that of BP decoding. Coming back to
the comparisons between the $\ell_1$ penalty and the $\ell_2$ penalty, we observe
that decoding using the $\ell_2$ penalty has a slightly better error performance and time efficiency. 

To summarize, \pPD displays best error performance for the codes and SNRs we simulate.
What is more, \pPD is also the {fastest} among the decoders we simulated\footnote{Measured by our implementations executed on the same CPU.}. These make \pPD a very strong competitor among the decoders studied in this work.

\subsubsection{Choice of parameters}
\label{sec.sim.parameter}
There are many parameters that need to be chosen for \pPD. As a reminder, these parameters include:
the penalty coefficient $\alpha$, the ADMM augmented Lagrangian weight $\mu$, 
ADMM ending tolerance $\epsilon$, the ADMM maximum number of iterations $T_{\max}$,
and the ADMM over-relaxation parameter $\rho$ (cf.~\cite{barman2013decomposition}). The optimal parameter settings
depend on both the code and the SNR it operates at. Exhaustively 
search for all possible combinations of parameters for all SNRs is infeasible and
is not the focus of this paper. The parameter choice principles learned in this paper 
are already satisfactory in order to obtain good performances as demonstrated in Section~\ref{sec.sim.wer}.
In this section, we fix some parameters
and then search for the rest.
In particular, we search for optimal parameters at low SNRs. We show in Section~\ref{sec.sim.instanton} that the optimal value for $\alpha$ should be small for high SNRs.

We first study WER as a function of $\alpha_i, i \in\{1,2\}$ for
the respective updates steps of~\eqref{eq.lpl1penalty}
and \eqref{eq.lpl2penalty} in
Fig.~\ref{fig.sim_margulis_parameteralpha_snr_wer} 
and
Fig.~\ref{fig.sim_margulis_parameteralpha_other_wer}. 
 We fix the ADMM parameters $\mu = 3$ and $\epsilon = 10^{-5}$.
  For each data point, we collect more than
$200$ errors. 
In Fig.~\ref{fig.sim_margulis_parameteralpha_snr_wer}, we fix the maximum number
of iterations $T_{\max} = 1000$. In Fig.~\ref{fig.sim_margulis_parameteralpha_other_wer}, we only plot data for 
\pPD with an $\ell_2$ penalty for conciseness. The respective behaviors for \pPD with an $\ell_1$ penalty is 
similar to Fig.~\ref{fig.sim_margulis_parameteralpha_other_wer}.

\begin{figure}
\psfrag{&&&&&L1PenaliedDecoderSNR16}{\scalebox{.6}{$\ell_1$ penalty, $E_b/N_0 = 1.6\dB$}}
\psfrag{&&&&L2PenaliedDecoderSNR16}{\scalebox{.6}{$\ell_2$ penalty, $E_b/N_0 = 1.6\dB$}}
\psfrag{&&&&L1PenaliedDecoderSNR19}{\scalebox{.6}{$\ell_1$ penalty, $E_b/N_0 = 1.9\dB$}}
\psfrag{&&&&L2PenaliedDecoderSNR19}{\scalebox{.6}{$\ell_2$ penalty, $E_b/N_0 = 1.9\dB$}}
\psfrag{&Alpha}{\scalebox{.6}{$\alpha_i$; $i = \{1,2\}$}}
\psfrag{&WER}{\scalebox{.6}{WER}}

\psfrag{0}{\scriptsize{$0$}}
\psfrag{1}{\scriptsize{$1$}}
\psfrag{2}{\scriptsize{$2$}}
\psfrag{3}{\scriptsize{$3$}}
\psfrag{4}{\scriptsize{$4$}}
\psfrag{5}{\scriptsize{$5$}}
\psfrag{6}{\scriptsize{$6$}}
\psfrag{10}{\scriptsize{\hspace{-1mm}$10$}}
\psfrag{-1}{\tiny{\hspace{-1mm}$-1$}}
\psfrag{e0}{\tiny{\hspace{-1mm}$0$}}
    \begin{center}
    \includegraphics[width=3.5in]{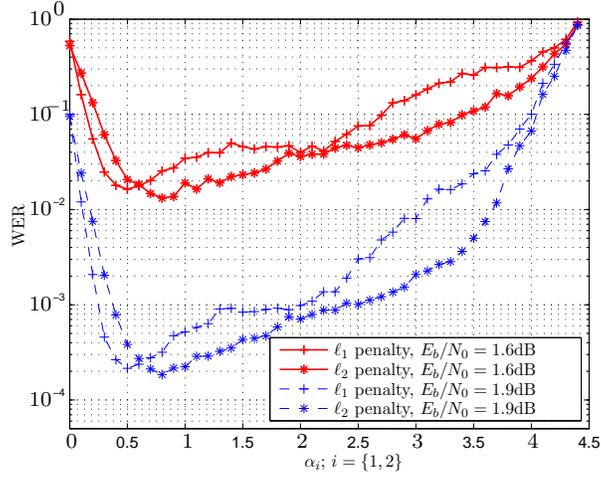}
    \end{center}
    \caption{Word-error-rate (WER) of the $[2640,1320]$ ``Margulis'' code for the AWGN channel plotted as a function of the choice of $\alpha_i$; $E_b/N_0 = 1.6\dB$ and $1.9\dB$.}
    \label{fig.sim_margulis_parameteralpha_snr_wer}
\end{figure}

\begin{figure}
\psfrag{&&&&Tmax200Overrel1}{\scalebox{.6}{$T_{\max} = 200$, $\rho = 1$}}
\psfrag{&&&&Tmax200Overrel19}{\scalebox{.6}{$T_{\max} = 200$, $\rho = 1.9$}}
\psfrag{&&&&Tmax1000Overrel1}{\scalebox{.6}{$T_{\max} = 1000$, $\rho = 1$}}
\psfrag{&&&&&Tmax1000Overrel19}{\scalebox{.6}{$T_{\max} = 1000$, $\rho = 1.9$}}
\psfrag{&Alpha}{\scalebox{.6}{$\alpha_i$; $i = \{1,2\}$}}
\psfrag{&WER}{\scalebox{.6}{WER}}

    \begin{center}
    \includegraphics[width=3.5in]{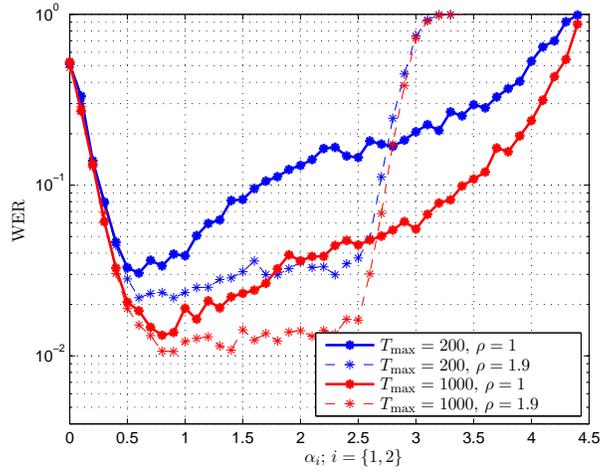}
    \end{center}
    \caption{Word-error-rate (WER) of the $[2640,1320]$ ``Margulis'' code for the AWGN channel plotted as a function of the choice of $\alpha_i$; $E_b/N_0 = 1.6\dB$ with different maximum number of iterations $T_{\max}$ and over-relaxation parameter $\rho$. \pPD with an $\ell_2$ penalty.}
    \label{fig.sim_margulis_parameteralpha_other_wer}
\end{figure}

We make the following observations from Fig.~\ref{fig.sim_margulis_parameteralpha_snr_wer} and Fig.~\ref{fig.sim_margulis_parameteralpha_other_wer}: First, the WER drops quickly as we start to penalize
fractional solutions (increasing the $\alpha_i$ above
zero). But the WER begins to increase again
when we penalize too heavily, e.g., when we set $\alpha_1>0.5$ or $\alpha_2>0.8$. 
The WERs then slowly increase up to unity. The
globally optimal parameter choices are $\alpha_1^* = 0.5, \alpha_2^* =
0.8$.\footnote{When $\rho = 1.9$, $\alpha_1^*$ becomes $0.6$. The data is not shown here but it is similar to the case in Fig.~\ref{fig.sim_margulis_parameteralpha_other_wer}} 
Second, the optimal parameters are almost the same for $E_b/N_0 = 1.6$ and $1.9\dB$.
Recall that in Fig.~\ref{fig.sim_margulis}, we used the same $\alpha$'s for all SNRs and 
received consistent gains throughout all SNRs. Although this simplifies practical implementation,
one needs to be careful at high SNRs. We defer the detailed discussions to Section~\ref{sec.sim.instanton}.
Third, although different values for
$T_{\max}$ give almost the same performance for LP decoding (see
points at $\alpha_i = 0$), $T_{\max}$ does affect the performance when
a penalty term is added. We see from
Fig.~\ref{fig.sim_margulis_parameteralpha_other_wer} that when more iterations
are allowed, the WER decreases substantially. This also applies to the
over-relaxation parameter, i.e., it affects the WER for \pPD but does not affect the WER for LP decoding. 

In Fig.~\ref{fig.sim_margulis_parametermu_other_wer}, we focus on the augmented Lagrangian weight parameter $\mu$ and study
the effects of its choice on WER. We observe that WER has a weak dependency on $\mu$. The optimal $\mu$s for
$T_{\max} = 200$ and $T_{\max} = 1000$ are $\mu = 3$ and $\mu = 4$, respectively. In fact, even
if the decoder is not operating with the optimal $\mu$, the performance degradation is small. We
also note that the decoding time changes little for $\mu \in [2,6]$ (data not
shown here). Thus we conclude that parameter $\mu$ has a smaller impact on decoder performance as compared to the penalty
parameter $\alpha$. This property provides rich choices for $\mu$ in practice.

\begin{figure}
\psfrag{&&&&&&L1Tmax200Overrel1}{\scalebox{.6}{$T_{\max} = 200$, $\ell_1$ penalty}}
\psfrag{&&&&&&L2Tmax200Overrel1}{\scalebox{.6}{$T_{\max} = 200$, $\ell_2$ penalty}}
\psfrag{&&&&&&L1Tmax1000Overrel1}{\scalebox{.6}{$T_{\max} = 1000$, $\ell_1$ penalty}}
\psfrag{&&&&&&L2Tmax1000Overrel1}{\scalebox{.6}{$T_{\max} = 1000$, $\ell_2$ penalty}}
\psfrag{&Mu}{\scalebox{.6}{$\mu$}}
\psfrag{&WER}{\scalebox{.6}{WER}}

    \begin{center}
    \includegraphics[width=3.5in]{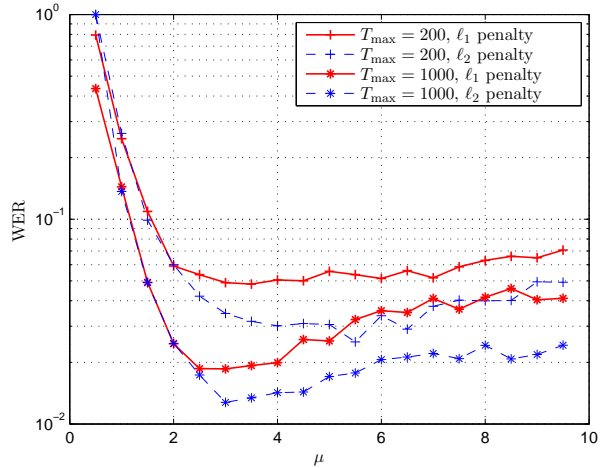}
    \end{center}
    \caption{Word-error-rate (WER) of the $[2640,1320]$ ``Margulis'' code for the AWGN channel plotted as a function of the choice of $\mu$; $E_b/N_0 = 1.6\dB$ with different maximum number of iterations $T_{\max}$.}
    \label{fig.sim_margulis_parametermu_other_wer}
\end{figure}

In Fig.~\ref{fig.sim_margulis_parameteralpha_snr_iter}, we study effects of the choice of penalty parameter $\alpha$
on decoder complexity. We make three remarks: First, the best parameter choice in terms of
WER agrees with the best parameter choice in terms of complexity. 
Second, the number of iterations for \pPD is much smaller than that for ADMM LP decoding (points where $\alpha_i = 0$ for $i = 1,2$).
Results in~\cite{barman2013decomposition} show that the decoding time for ADMM LP decoding 
is comparable to BP decoding. In recent works~\cite{zhang2013large, zhang2013efficient}, the authors proposed complexity improvements to the projection
onto parity polytope sub-routine, which can also be applied to \pPD directly. Thus we expect \pPD to be a strong competitor against BP decoding in terms of decoding speed. 
\begin{figure}
\psfrag{&&ADML1IterationsCorrectSNR16}{\scalebox{.6}{$\ell_1$ penalty, $E_b/N_0 = 1.6\dB$}}
\psfrag{&ADML2IterationsCorrectSNR16}{\scalebox{.6}{$\ell_2$ penalty, $E_b/N_0 = 1.6\dB$}}
\psfrag{&ADML1IterationsCorrectSNR19}{\scalebox{.6}{$\ell_1$ penalty, $E_b/N_0 = 1.9\dB$}}
\psfrag{&ADML2IterationsCorrectSNR19}{\scalebox{.6}{$\ell_2$ penalty, $E_b/N_0 = 1.9\dB$}}
\psfrag{&Alpha}{\scalebox{.6}{$\alpha_i$; $i = \{1,2\}$}}
\psfrag{&Iteration}{\scalebox{.6}{\# Iteration}}

    \begin{center}
    \includegraphics[width=3.5in]{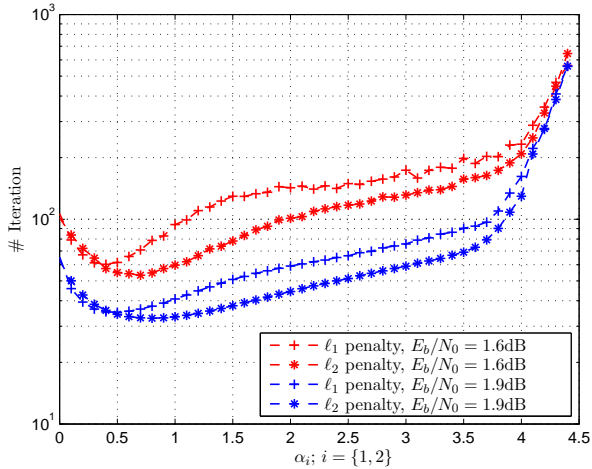}
    \end{center}
    \caption{Number of iterations of \textbf{correct decoding events} of the $[2640,1320]$ ``Margulis'' code for the AWGN channel plotted as a function of the choice of $\alpha_i$; $E_b/N_0 = 1.6\dB$ and $1.9\dB$.}
    \label{fig.sim_margulis_parameteralpha_snr_iter}
\end{figure}

We summarize the parameter choice methods learned in this section: 
\begin{itemize}
\item $\mu \in [3,5]$ are all good choices for both the $\ell_1$ and $\ell_2$ penalties;
\item Increasing $T_{\max}$ improves the WER performance. However it also decreases decoding speed. $T_{\max} = 200$ is large enough for good WER performances for any code studied in this paper;
\item  $\rho = 1.9$ is optimal for both WER and decoding speed;
\item We do not present data for parameter $\epsilon$, but we note that $\epsilon < 10^{-5}$ is sufficiently small for good WER performances;
\item Given the parameters above, $\alpha_1 = 0.6$ and $\alpha_2 = 0.8$ are good choices for $\ell_1$ and $\ell_2$ penalties respectively.
\end{itemize}

Although \pPD using non-optimal parameters can still outperform LP decoding, \pPD is more sensitive to parameter settings when compared with ADMM LP decoding (cf.~\cite{barman2013decomposition}). Therefore one should be careful in optimizing parameters when applying this decoder in practice.
\section{An instanton analysis for \pPD}
\label{sec.instanton}
We observe that \pPD achieves strikingly good SNR performance in the previous section. However the results we can obtain via simulations are limited. One important test \pPD needs to pass is the high SNRs error performance test. In this regime BP decoding suffers from error floors. Unlike LP decoding where convex optimization analyses can be applied (e.g. those presented in~\cite{feldman2007lp,daskalakis2008probabilistic,arora2012message,halabi2011lp}), \pPD is based on non-convex objective functions. Therefore it is infeasible to apply convex analyses directly. An instanton analysis (cf.~\cite{chilappagari2009instanton}), on the other hand, is generic and can be applied to any decoder. We first briefly review the instanton analysis framework and then propose an instanton search algorithm for \pPD.

\subsection{Review of the instanton concept}
In~\cite{stepanov2005diagnosis}, the authors introduced the instanton analysis for error correction codes. Instantons are the most probably noise configurations that cause decoding failures. In this analysis, the instanton information of a decoder, when decoding a particular code, is used to predict the slope of the WER curve at high SNRs (see also e.g.~\cite{stepanov2006instanton, chilappagari2009instanton}). 
The generality of an instanton analysis
makes it a good candidate for analyzing high SNR performance of \pPD. One necessary part of this analysis is an algorithm that identifies instantons. The genetic ``amoeba'' approach is used in~\cite{stepanov2005diagnosis}, for continuous channels. However specific algorithms that exploit the characteristics of a decoder can perform better. An instanton search algorithm for BP decoding in the AWGN channel is introduced in~\cite{stepanov2006instanton}. For LP decoding, there is a tight relationship between 
instantons and pseudocodewords~\cite{chertkov2008efficient}. Based on this connection,~\cite{chertkov2008efficient,chertkov2011polytope} proposed iterative algorithms to search for pseudocodewords for the AWGN channel. 

We first review the definition of instantons and then rephrase instanton search as a non-convex optimization problem.
\begin{definition}(cf. e.g.~\cite{chilappagari2009instanton})
Instantons are the most probable configurations of the channel noise that result in decoding failures.
\end{definition}

At high SNRs, the WER is dominated by the most probable instantons~\cite{stepanov2005diagnosis,stepanov2006instanton}. Searching for this instanton is equivalent to the following optimization problem: Let $\bfn^{\inst}$ denote the most probable instanton for a code and a decoder, then
\begin{equation*}
\bfn^{\inst} = \argmax_{\bfn:\text{ decoding fails with noise $\bfn$}} \Pr(\bfn).
\end{equation*}
The constraint set is non-convex in general. In the AWGN channel, this can be translated to finding the minimum norm noise configuration such that the decoder fails. That is,
\begin{equation}
\label{eq.instanton_search_objective}
\min_{\bfn} \; \|\bfn\|_2^2 
+ (C - \|\bfn\|_2^2 )  \mathbb{I}_{\{\text{decoding succeeds with noise $\bfn$}\}}
\end{equation}
where $\mathbb{I}_{\{\text{\textit{statement}\}}}$ is the indicator function that equals to $1$ if \textit{statement} is true and $0$ otherwise; $C > 0$ is a large enough constant so that all decoding failures attain less cost in this objective than any decoding success. 

This optimization perspective is illustrated by a cartoon in Fig.~\ref{fig.inst_cartoon}. We aim to search for the instantons that are closest to the origin. However we also build a barrier (the indicator function in~\eqref{eq.instanton_search_objective}) such that we only consider the vectors that cause decoding failures. In addition, the cartoon shows two important points: First, the feasible region is non-convex. Second, there will be multiple instantons. Therefore multiplicity information is also important. We discuss this in Section~\ref{subsec.trapping_set}. We also note that the region for which decoding fails may not be continuous as shown in the cartoon.
\begin{figure}
	\begin{center}
    \includegraphics[width=3.5in]{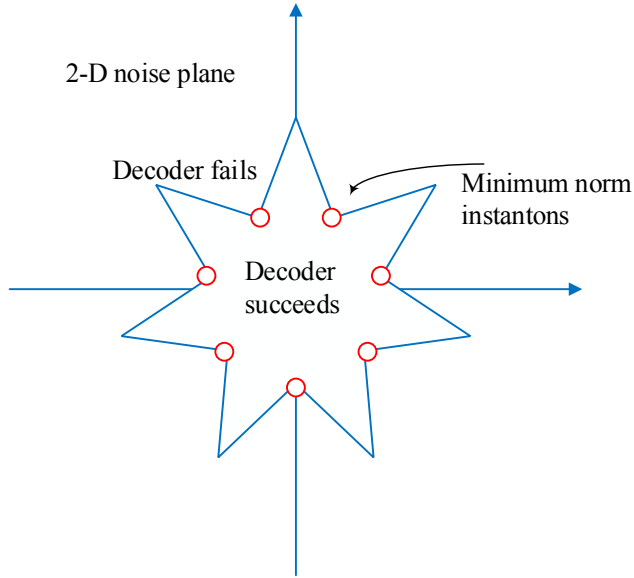}
    \end{center}
    \caption{A cartoon showing the idea of instanton.}
    \label{fig.inst_cartoon}
\end{figure}

\subsection{Instanton search algorithm for \pPD}
\label{subsection.instanton_search_algorithm}
We first introduce a specific algorithm for the ADMM penalized decoder denoted by \isapd (instanton search algorithm for the penalized decoder). Then we describe a generic algorithm, denoted by \isarf (instanton search algorithm, refining), which is used to refine the output from \isapd. We note that \isapd is already able to produce good instanton search results. However \isarf can still refine the results obtained by \isapd. This is demonstrated in Section~\ref{sec.sim.instanton}.
\subsubsection{Specific algorithm for \pPD}
In~\cite{chertkov2008efficient}, the authors propose an iterative algorithm to find the minimum weight pseudocodeword in the AWGN channel. For each iteration, the algorithm use an LP decoder to decode a noise vector that is calculated in the previous iteration to guarantee that LP decoding would fail. The output from the LP decoder is a pseudocodeword, which is used to calculate the minimum norm noise vector that can produce this failure. This noise vector is then used in the next iteration to be decoded by the LP decoder. The algorithm terminates when it converges or exceeds some maximum number of allowable iterations. This algorithm is executed repeatedly with random initializations in order to get statistics on instantons.

For \pPD, there is no analytical connection between decoder failures and instantons due to the fact that \pPD is trying to solve a non-convex problem. Whenever \pPD fails, the output is not necessarily a vertex of the fundamental polytope. However since both LP decoding and \pPD are minimizing an objective over the same constraint set, we can develop an iterative algorithm similar to~\cite{chertkov2008efficient}. Note that for
\pPD, we need a different approach to determine
a noise vector for the next iteration. Assuming that the all-zero codeword is transmitted and the decoder's output is $\bfomega$ for the current iteration, we solve the following optimization problem:
\begin{equation}
\begin{aligned}
\label{eq.instanton_opt}
\min \;  &\|\bfn\|_2^2 \\ 
\suchthat \;  &\bfgamma(\bfn)^T \bfomega  + \sum_i g(\omega_i) \leq  N g(0)
\end{aligned}
\end{equation}
where $\bfgamma(\bfn)$ is the LLR for noise $\bfn$ and $N$ is the block length.
This problem is solving for the smallest norm noise vector that confuses $\bfzero$ and $\bfomega$.
Using standard Lagrange multiplier technique, we can obtain the solution
\begin{equation}
\label{eq.instanton_search_vector}
\bfn^* = \bfomega \left(\frac{1}{\| \bfomega \|_2^2} \left[\| \bfomega \|_1 + \frac{ \sigma^2}{2}\left(N g(0) -\sum_i g(\omega_i)  \right)\right]\right),
\end{equation}
where $g(\cdot)$ is the same function in~\eqref{eq.generalpenalty} and can be replaced with the specific penalty used. Note that $\bfx^*$ is an approximation to the minimum noise vector causing \pPD. \pPD can still fail on a noise vector with smaller norm. One needs to search for the smallest norm noise vector that has the same direction as $\bfx^*$. This is captured in Algorithm~\ref{Algorithm.instanton}. Note that this algorithm may not converge. And, similar to~\cite{chertkov2008efficient}, different initializations produce different instantons. One should expect to run the algorithm multiple times to obtain statistics of instantons for \pPD.
\begin{algorithm}
\caption{Instanton Search Algorithm for \pPD (\isapd)}
\label{Algorithm.instanton}
\begin{algorithmic}[1]
\STATE Initialize a noise vector $\bfn^{0}$ sufficiently large that \pPD fails and outputs a vector $\bfomega^0$. 
\STATE Choose the maximum number of iterations allowed $T$. Choose a tolerance $\epsilon > 0$.
\FORALL{ $ k = 1,\dots, T $}
\STATE Let $\bfw$ be the solution of problem~\eqref{eq.instanton_opt} given $\bfomega^{k-1}$.
\STATE  Search for the smallest scaling factor $a>0$ such that decoder fails with noise $a\bfw$.
\STATE  Assign $\bfn^k \leftarrow a \bfw$.
\STATE Let $\bfomega^k$ be the output of decoder when decoding the noise vector $\bfn^k$.
\STATE Store $\bfn^k$ if it is the minimum norm noise vector from all previous trials.
\IF {$\| \bfn^k - \bfn^{k-1} \|_2 \leq \epsilon$}
    \STATE  Stop loop.
\ENDIF
\ENDFOR
\end{algorithmic}
\end{algorithm} 

\subsubsection{Using a generic instanton search algorithm as a refining step}
\label{subsec.nesterov}
 Although the ``amoeba'' minimization algorithm can be applied to any decoder, the cost of using this algorithm is huge. The algorithm not only converges slowly, but also requires memory that scales quadratically with respect to block length\footnote{The algorithm needs to save $n+1$ vectors of length $n$, which results in $O(n^2)$ spatial complexity.}. Thus it is infeasible to apply this algorithm to large block length codes. 

In this paper, we base our instanton search algorithm, \isarf, on Nesterov's random gradient free algorithm (cf.~\cite{nesterov2011random}). We denote by $f(\cdot)$ the objective function to minimize. In short, the algorithm in~\cite{nesterov2011random} starts at a random initial point $\bfx_0$. Fix a parameter $\eta > 0$ and pick a sequence of positive steps $\{h_k\}_{k\geq 0}$. For each iteration $k>0$, a vector $\bfu$ is generate following Gaussian distribution with correlation operator $\bfB^{-1}$. After that, a \emph{random gradient-free oracle} is computed:
\begin{equation*}
\phi_{\eta}(\bfx_k) = \frac{f(\bfx_k + \eta \bfu) - f(\bfx_k)}{\eta} \cdot \bfB \bfu.
\end{equation*}
Then the next step is determined by 
\begin{equation*}
\bfx_{k+1} = \bfx_k - h_k \bfB^{-1} \phi_{\eta}(\bfx_k)
\end{equation*}
This algorithm only requires the two most recent vectors thus consumes much less memory compared to the ``amoeba'' minimization. 

Instead of directly applying this method to the objective function in~\eqref{eq.instanton_search_objective}, we make three important tweaks: 
First, we make the objective function have a non-zero gradient when the decoder succeeds:
\begin{equation}
\label{eq.instanton_search_objective_with_cap}
\begin{split}
\min_{\bfn} \; (1-\mathbb{I}_{\{\text{decoding succeeds with noise $\bfn$}\}}) \|\bfn\|_2^2 \\
+ C \mathbb{I}_{\{\text{decoding succeeds with noise $\bfn$}\}} \left( 1 - \|\bfn\|_2^2\right).
\end{split}
\end{equation}
In other words, for each iteration, we run \pPD using the current noise vector and evaluate~\eqref{eq.instanton_search_objective_with_cap}.

Second, we leverage the knowledge learned from \isapd. Similar to the case of ``amoeba'' minimization~\cite{chilappagari2009instanton}, initialization of \isarf is crucial in order to obtain small norm instantons. In this paper, we use outputs from Algorithm~\ref{Algorithm.instanton} as the initialization. As we will show in Section~\ref{sec.sim.instanton}, the instanton search results from \isapd tend to have only a few non-zero entries and the rest of the entries are zero. We use this knowledge in our algorithm and only select the random Gaussian search direction in Nesterov's method to have support on the coordinates for which the entries are non-zero entries. 

Third, the indicator function in~\eqref{eq.instanton_search_objective_with_cap} may create a huge step size in \isarf depending on the constant $C$. Thus we restrict the search step size for each iteration by linearly shrinking the vector $h_k \bfB^{-1} \phi_{\eta}(\bfx_k)$ if its $\ell_2$-norm exceeds some threshold. 

\subsection{Numerical results of the instanton analysis}
\label{sec.sim.instanton}
We now apply the proposed algorithm to the $[155,64]$ Tanner code~\cite{tanner2001class} and a $[1057,813]$ LDPC code from~\cite{Database}. For both codes, we obtain instanton information for different penalty coefficients $\alpha_1$ and $\alpha_2$. We fix parameters (cf. Section~\ref{sec.numerical_penalty}) as follows:
$\mu = 3$, $T_{\max} = 100$, $\epsilon = 10^{-5}$, $\rho = 1.9$.\footnote{We also fixed the channel standard deviation parameter to be $0.5$ for the $[155,64]$ Tanner code and $0.45$ for the $[1057,813]$ code. This parameter is used to calcluate the log-likelihood ratios $\bfgamma$ in \pPD.} For each data point, we randomly generate $K = 1000$ initial noise vectors with power large enough to cause decoding failures and apply \isapd followed by \isarf. For \isarf, we use $C = 20000$ (cf.~\eqref{eq.instanton_search_objective}). For \isarf, we let $\bfu_k$ be an Gaussian vector with the identity covariance matrix ($\bfB = \eye$), $\eta = 10^{-10}$ and $h_k = \frac{1}{40000}$. Moreover, we restrict each search step size to have $\ell_2$ norm less than or equal to 1.

In Fig.~\ref{fig.inst_norm}, we plot the minimum instanton norm as a function of penalty coefficient $\alpha_2$. For each $\alpha_2$, we sort the $K$ instantons we find, denoted by $\bfn_i$, $i = 1,...,K$, such that $\|\bfn_{1}\|_2^2 \leq \|\bfn_{2}\|_2^2 \leq \cdots \leq \|\bfn_{K}\|_2^2$. We then plot the norm $\|\bfn_{\phi}\|_2^2$ where $\phi = K/100=10$ and label it ``$1\%$ minimum instanton'' in Fig.~\ref{fig.inst_norm}. In order to compare with BP and LP decoding, we plot the effective distance for both decoders. Note that the LP pseudo-distance is actually the norm of the most probable noise that leads to decoding failure (see also~\cite[Sec.~6]{vontobel2005graph}). Results for BP decoding are obtained from~\cite{stepanov2011instantons}.

We make the following observations: First, there is a clear trend that as the penalty increases, the minimum instanton norm decreases. This implies that the high-SNR performance of \pPD decreases as the penalty coefficient increases. Second, we observe that the fluctuation of the minimum norm curves increases as $\alpha_2$ increases. These fluctuations are partially compensated by \isarf. We observe that \isarf makes some refinements at points at $\alpha = 0.5,0.8$. It makes a significant refinement at $\alpha = 1.8$. This is important because it suggests that the results obtained by \isapd for large $\alpha_2$ are not accurate especially for $\alpha_2 \geq 1.6$. Therefore, a curve-fitted worst case instanton norm (shown in in Fig.~\ref{fig.inst_sim}) should predict WER more accurately for large $\alpha_2$'s compared to raw data. Thus we suggest using the curve-fitted results. E.g. for $\alpha = 2$, results suggest that the minimum instanton norm is $13.8$. We show in Fig.~\ref{fig.inst_sim} that this result predicts the WER slope well. 

In Fig.~\ref{fig.inst_sim}, we plot the WER asymptotics predicted by the minimum norm instanton as well as simulation results. We use $\alpha_2 = 2$ in our simulation because this value is \emph{optimal} in terms of WER for $\SNR = 1.3$ and $1.65$ for this code. All data points are based on more than 100 decoding errors. We make three observations: First, the asymptotics predicted by the minimum instanton norms and the simulated WER curves are almost parallel for each decoder at high SNRs. 
Second, there are two crossover points in Fig.~\ref{fig.inst_sim}. The first one occurs near $\SNR = 1.75$ where LP decoding outperforms BP decoding. The second one occurs near $\SNR = 2.3 $ where LP decoding outperforms \pPD with an $\ell_2$ penalty. This matches our expectation from the instanton results in Fig.~\ref{fig.inst_norm}. Third, the asymptotic curves approximate the slope of the WER. If we were to approximate the actual WER, we would need to compute the multiplicity and the curvature factor of the instantons (cf.~\cite{chilappagari2009instanton}). We discuss the multiplicity in Section~\ref{subsec.trapping_set}. 

In Fig.~\ref{fig.inst_norm_1057}, we plot the instanton profile for the $[1057,813]$ code for \pPD with both $\ell_1$ and $\ell_2$ penalties. We first use the algorithm in~\cite{chertkov2008efficient} to obtain statistics for LP decoding. Based on all instantons we observed, the minimum norm instanton for LP decoding is an actual codeword. This implies that LP decoding is as good as ML decoding at high SNRs for this code, which matches the simulation results obtained in~\cite{barman2013decomposition}. We make three comments on Fig.~\ref{fig.inst_norm_1057}. First, the algorithm in~\cite{chertkov2008efficient} is not guaranteed to find the minimum weight pseudocodeword. If the weight were actually $7$ for LP decoding then our results would suggest that we can push the penalty coefficient up to $1$ without losing much performance at high SNRs. Second, we observe significant fluctuations in terms of instanton norms. However the one percent instanton norm curves for both penalties decrease monotonically as the penalty coefficients $\alpha_1$ and $\alpha_2$ increase. The fluctuation is thus due to limited simulation time. Third, \pPD with an $\ell_1$ penalty does worse than \pPD with an $\ell_2$ penalty for a given parameter $\alpha_i$. However, the results obtained in~\cite{liu2012suppressing} suggest that the optimal $\alpha_1$ is less than the optimal $\alpha_2$. This implies that both penalties are comparable in terms of instanton norms when each decoder uses its optimal parameter setting. 

These results are important in practice. First, although \pPD is observed to outperform LP decoding in~\cite{liu2012suppressing}, it is not as good as LP decoding asymptotically in SNR for nonzero penalty coefficients. Second, in practice, we need to adjust the penalty coefficient as a function of the SNR the code operates at. In particular, we can expect that small penalty coefficients produce better error rate performances at high SNRs, as is shown in Fig.~\ref{fig.inst_sim}.

\begin{figure}
\psfrag{&A}{\scalebox{.6}{\hspace{-1cm} penalty coefficient $\alpha_2$}}
\psfrag{&B}{\scalebox{.6}{\hspace{-1cm} Instanton norm}}
\psfrag{&LP}{\scalebox{.6}{LP instanton (pseudo-distance)~\cite{chertkov2008efficient}}}
\psfrag{&OnePercentCurveL2padpadpadpadpadpad}{\scalebox{.6}{$1\%$ minimum instanton}}
\psfrag{&MinimumCurveL2}{\scalebox{.6}{Minimum instanton norm from \isapd}}
\psfrag{&MinimumRefineL2}{\scalebox{.6}{Refined results from \isarf}}
\psfrag{&BP}{\scalebox{.6}{BP decoding, 100 iterations~\cite{stepanov2011instantons}}}
\psfrag{&curve-fitted}{\scalebox{.6}{Curve-fitted instanton norm}}
	\begin{center}
    \includegraphics[width=3.5in]{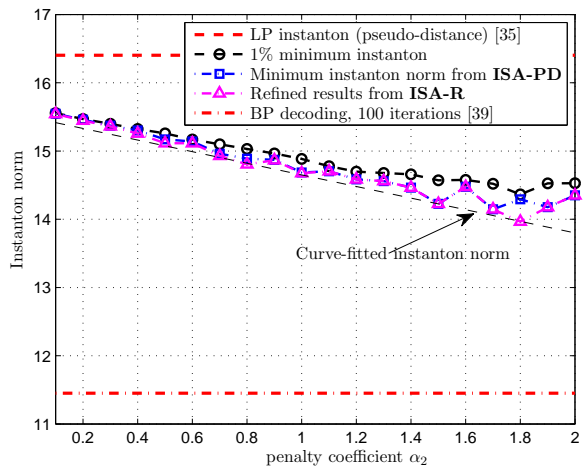}
    \end{center}
    \caption{Instanton norm as a function of penalty coefficient for the
      $[155,64]$ code when decoded by \pPD with $\ell_2$ penalty.}
    \label{fig.inst_norm}
\end{figure}

\begin{figure}
\psfrag{&ADMMLP}{\scalebox{.6}{ADMM LP decoding}}
\psfrag{&BP}{\scalebox{.6}{Sum-product BP decoding, 100 iterations}}
\psfrag{&ADMML2}{\scalebox{.6}{\pPD with $\ell_2$ penalty, 100 iterations}}
\psfrag{&InstantonPredictionForBP}{\scalebox{.6}{Asymptote, BP}}
\psfrag{&InstantonPredictionForL2-padpadpadpadpadpa}{\scalebox{.6}{Asymptote, \pPD with $\alpha_2 = 2$}}
\psfrag{&InstantonPredictionForLP}{\scalebox{.6}{Asymptote, LP}}
\psfrag{&SNR}{\scalebox{.6}{\hspace{-1mm}SNR}}
\psfrag{&WER}{\scalebox{.6}{WER}}
\begin{center}
\includegraphics[width=3.5in]{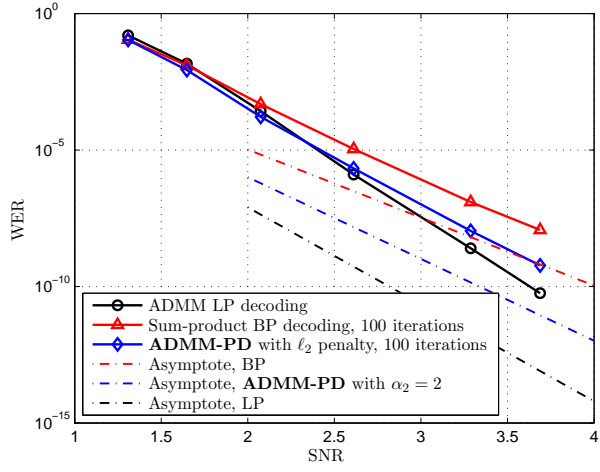}
\end{center}
\caption{Comparisons of simulation results and asymptotic slopes predicted by instantons/pseudocodewords for the $[155,64]$ code. WER is plotted as a function of SNR (not in (dB)). Each asymptote is $\sim\exp(-\|\bfn^{\inst}\|_2^2s^2/2)$, where $s^2$ is SNR. The value of $\|\bfn^{\inst}\|_2^2$ for BP, \pPD and LP are $11.48$, $13.8$ and $16.35$, respectively.}
\label{fig.inst_sim}
\end{figure}

\begin{figure}
\psfrag{&A}{\scalebox{.6}{\hspace{-2cm} penalty coefficient $\alpha_i$, $i = 1,2$}}
\psfrag{&B}{\scalebox{.6}{\hspace{-1cm} Instanton norm}}
\psfrag{&LP}{\scalebox{.6}{LP pseudo-distance}}
\psfrag{&OnePercentCurveL1}{\scalebox{.6}{$1\%$ minimum instanton, $\ell_1$ penalty}}
\psfrag{&MinimumCurveL1}{\scalebox{.6}{Minimum instanton norm, $\ell_1$ penalty}}
\psfrag{&OnePercentCurveL2padpadpadpadpad}{\scalebox{.6}{$1\%$ minimum instanton, $\ell_2$ penalty}}
\psfrag{&MinimumCurveL2}{\scalebox{.6}{Minimum instanton norm, $\ell_2$ penalty}}
	\begin{center}
    \includegraphics[width=3.5in]{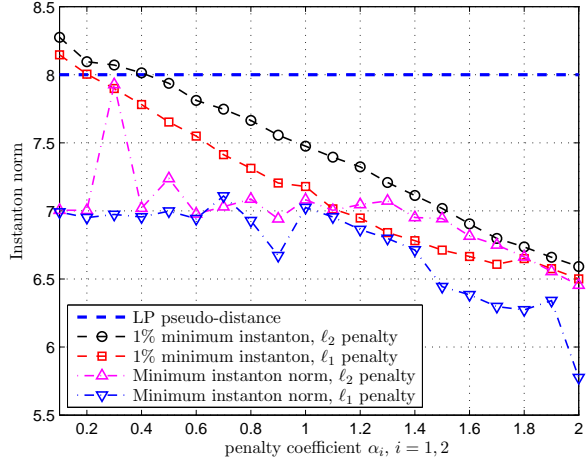}
    \end{center}
    \caption{Instanton norm as a function of penalty coefficient for the
      $[1057,813]$ code when decoded by \pPD.}
    \label{fig.inst_norm_1057}
\end{figure}

\subsection{Weaknesses of \pPD}
\label{subsec.trapping_set}
It is noticeable from Fig.~\ref{fig.inst_norm_1057} that the minimum instanton norm (not the $1\%$) stays almost constant for some range of $\alpha$. This suggests that there might be some common structures underlying these instantons, which motivates us to study the  patterns of these noise vectors. 

We observe that for all minimum norm instantons, there are only a few nonzero entries. Moreover, every instanton we studied has support on trapping set structures widely studied for BP decoders. For the $[155,64]$ code, we observe instantons with non-zeros entries on $(5,3)$ (plotted in Fig.~\ref{fig.trapping_set}) and $(6,4)$ trapping sets. For the $[1057,814]$ code, we observe instantons with non-zeros entries on $(5,1)$ and $(7,1)$ trapping sets. And the $(7,1)$ trapping set contributes to most of the norm $7$ instantons in Fig.~\ref{fig.inst_norm_1057}.

We note two aspects where these results are useful. First, the results imply that we can leverage trapping set knowledge in the literature to compute the multiplicity of a particular instanton. Second, one can expect that code designs that alleviate trapping set effect should also benefit \pPD. 

\begin{figure}
\begin{center}
\includegraphics[width=6.1cm]{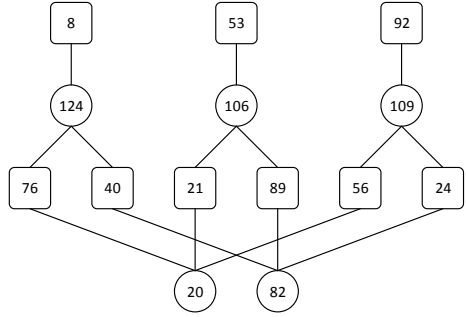}
\end{center}
\caption{The support of an instanton for the $[155,64]$ code showing a $(5,3)$ trapping set, each node indexed by the variable or check number.}
\label{fig.trapping_set}
\end{figure}

\subsection{Remarks}
We now briefly discuss our intuition gained from these results. We note that the penalized decoder aims to penalize fractional solutions by adding penalties to the objective on a symbol-by-symbol basis. As a result, the decoder more easily gets stuck at some local configurations such as trapping sets. Meanwhile, increasing the penalty coefficient can decrease the significance of channel evidence during the decoding process. Therefore a small noise can create a bias at the first iteration of decoding that is hard to reverse. So the instanton norm tends to be small for large penalty coefficients.

Next, we note that in \isapd, the search direction per iteration is the same as the algorithm proposed in~\cite{chertkov2008efficient}. However for \isapd, we need to identify the correct step size along this direction. Nevertheless, we observe from our simulation that \eqref{eq.instanton_search_vector} is a good approximation.

In addition, we note that we would need the curvature information in order to approximate the exact WER for the AWGN channel (cf.~\cite{chilappagari2009instanton}). This question has not been studied in the literature and is beyond the scope of this paper. Thus we leave it as a future research direction.

Finally, it is desirable to obtain instanton information for the BSC. However this is not trivial and is beyond the scope of this paper. The reason is that searching for instantons in the BSC can be formulated as a non-convex integer program with feasible set $\{0,1\}^{\blocklength}$ and hence is in general intractable unless some knowledge on error patterns are pre-assumed. With our findings that the failures of \pPD is related to trapping set, a natural next step is to study the dynamics of \pPD in trapping sets with the assumption that the underlying channel is the BSC. We leave instanton searching for \pPD in the BSC an open problem.
\section{Reweighted LP decoding}
In previous sections, we show that \pPD is able to improve the low SNR performance of LP decoding. In addition, \pPD also has good high SNR behaviors. However, unlike LP decoding, it is hard to prove error correction guarantees for \pPD. In this section, we introduce a reweighted LP decoding algorithm derived from \pPD with an $\ell_1$ penalty and prove results on error correction guarantees. We use \pPLP ($\ell_1$ penalized decoding) to denote \pPD with an $\ell_1$ penalty and \pRLP (reweighted LP decoding) to denote our reweighted LP algorithm. Note that this reweighted LP is different from that in~\cite{khajehnejad2012reweighted} (cf. \eqref{eq.rlp_kdhb}). As a reminder, a comparison is listed in Table~\ref{tab.rlp_comparison}.
\begin{table}
\centering
  \begin{tabular}{| c | c | c |}
  \hline
   & \pRLP & \pRLPKDHVB\\
    \hline\hline
    \# rounds of reweighting & $\geq 2$ & 2 \\ \hline
    recovery guarantees & Yes\footnotemark & Yes \\ \hline
    BSC & Yes & Yes \\ \hline
    AWGN \& general MBISO & Yes & No\\
    \hline
  \end{tabular}
  \caption{Comparisons between \pRLP and \pRLPKDHVB.}
\label{tab.rlp_comparison}
\end{table}
\footnotetext{This guarantee is for the two-round \pRLP in the BSC, which has the same assumptions as~\cite{khajehnejad2012reweighted}.}
\subsection{A reweighted LP approximation for \pPLP}
\label{sec.reweightedLP}
The general reweighted LP scheme is well known in the compressed
sensing community~\cite{candes2008enhancing}. We are interested
in reweighted LPs because it approximates a nonlinear program using
a sequence of linear programs thereby bringing a non-convex problem to a sequence of convex
problems. We show that the two-round reweighted LP for problem
(\ref{eq.lpl1penalty}) has theoretical guarantees. That is, this
reweighted LP decoding yields improved recovery threshold than LP decoding
for bit flipping channels (cf.~\cite{feldman2007lp}).

We briefly summarize the reweighted LP idea
in~\cite{candes2008enhancing}. Consider optimization problem of the
following form:
\begin{equation*}
\min_{\bfv}\quad \phi(\bfv) \qquad \suchthat\quad \bfv \in \mathbb{C},
\end{equation*}
where $\mathbb{C}$ is a convex set. The reweighted LP scheme start
from a feasible solution $\bfv^{(0)}$ and {iteratively} solves
\begin{equation}
\label{eq.general_rlp}
\begin{aligned}
\bfv^{(k+1)} = &\argmin_{\bfv}\quad \phi(\bfv^{(k)})+ \nabla \phi(\bfv^{(k)})\left(\bfv - \bfv^{(k)}\right)\\
 &\suchthat\quad \bfv \in \mathbb{C},
 \end{aligned}
\end{equation}
where $k$ denotes the iterate. 

Applying~\eqref{eq.general_rlp} to \pPLP (cf.~\eqref{eq.lpl1penalty}), we get the following reweighted LP algorithm \pRLP. Note that we drop the subscript $\alpha_1$ in \eqref{eq.lpl1penalty} and use only $\alpha$ for simplicity. 
 \vspace{0.2cm}\\
\fbox{
 \addtolength{\linewidth}{-4\fboxsep}%
 \addtolength{\linewidth}{-4\fboxrule}%
 \begin{minipage}{\linewidth}
\begin{equation}
\label{eq.general_rlp_decoding}
\begin{aligned}
\pRLP \qquad \min & \quad\bfgamma^{(k)}(\alpha)^{T} \bfx\\
 \suchthat &\quad  \bfP_j \bfx \in \PP_{d_j},\forall j\in\cJ.
 \end{aligned}
\end{equation}
where $\bfx^{(0)}$ is some initialization and
\begin{equation*}
\bfgamma^{(k)}(\alpha)_i = \bfgamma_i  - \alpha \sgn(x^{(k-1)}_i - 0.5),\forall i
\end{equation*}
\end{minipage}
}\vspace{0.2cm}

We consider a simple case for our decoding problem: we consider just
the two-round version of \pRLP. That is, we first solve \pLP and
initialize $\bfx^{(0)}$ to be the pseudocodeword $\bfxpseudo$, then we solve
for $\bfx^{(1)}$. This scheme is summarized in
Algorithm~\ref{Algorithm:reweightedLP}. From a theoretical point of
view, we show that this two-round \pRLP has a better recovery threshold
compared to \pLP. From a practical point of view, the decoding
algorithm should be fast in order to perform real time decoding. The
cost in terms of efficiency for using multiple-round reweighting could
be unacceptable. 
\begin{algorithm}
\caption{Two-round \pRLP}
\label{Algorithm:reweightedLP}
\begin{algorithmic}[1]
\STATE Solve \pLP.
\IF {The solution $\bfx^{*}$ is an integral solution}
    \STATE  Stop. {\bf return} $\bfx^{*}$.
\ELSE
    \STATE  Save pseudocodeword $\bfxpseudo = \bfx^{*}$ and proceed.
\ENDIF
\FORALL{ $ i \in \mathcal{I} $}
\STATE Update 
$\gamma^\prime_i \leftarrow \gamma_i - \alpha \sgn(x^{(p)}_i - 0.5)$
\ENDFOR
\STATE Solve
\begin{equation*}
\min  \quad\bfgamma^{\prime T} \bfx \qquad
\suchthat  \quad\bfP_j \bfx \in \PP_{d_j}. 
\end{equation*}
\end{algorithmic}
\end{algorithm}

\begin{theorem}
\label{thm.reweightedlpcwindependence}
Decoding failures for Algorithm~\ref{Algorithm:reweightedLP} are independent of the codeword that was transmitted.
\end{theorem}
\begin{proof}
See Appendix~\ref{appendix.cwindep_proof_reweightedlp}.
\end{proof}
\subsection{Theoretical guarantees of two-round \pRLP}
\label{subsec.rlp_thm}
The second LP problem in Algorithm~\ref{Algorithm:reweightedLP} is
\begin{equation}
\label{eq.reweightedlp}
\begin{split}
\min & \quad\bfgamma^{\prime T} \bfx \\
\suchthat & \quad\bfP_j \bfx \in \PP_{d_j}, 
\end{split}
\end{equation}
where $\gamma^\prime_i = \gamma_i - \alpha \sgn(x^{(p)}_i - 0.5)$ and
$\alpha>0$ is a constant. Using lemmas introduced
in~\cite{khajehnejad2012reweighted}, we show that this scheme corrects
more errors than \pLP.

\begin{lemma}(Lemma 26 in \cite{vontobel2005graph})
The fundamental cone $\fundcone := \fundcone(\bfH)$ of parity check matrix $\bfH$ is the set of all vectors $\bfw \in \mathbb{R}^n$ that satisfy
\begin{align*}
&\w_i \geq 0,  &\forall i\in\cI,\\
&\w_i \leq \sum_{i^\prime\in\Ne_c(j) \setminus i }\w_{i^\prime}, &\forall j\in\cJ,\forall i \in \cI.
\end{align*}
\end{lemma}
\begin{definition}(Fundamental Cone Property~\cite{khajehnejad2012reweighted})
Let $\F\subset \{1,2,\dots,n\}$ and $\C\geq 1$ be fixed. A code $\codebook$ with parity check matrix $\bfH$ is said to have the fundamental cone property $FCP(\F,\C)$, if for every nonzero vector $\w \in \fundcone$ the following holds:
\begin{equation*}
\C\Vert\w_{\F}\Vert_1 < \Vert\w_{\F^c} \Vert.
\end{equation*}
If for every index set $\F$ of size $k$, $\codebook$ has the $FCP(\F,\C)$, then we say that $\codebook$ has $FCP(k,\C)$.
\end{definition}

The key insight of $FCP$ is that when $\C>1$, even if \pLP fails, the pseudocodewords still have the nice property described by Lemma~\ref{lemma.FCPl1bound} below. This turns out to be crucial for reweighted LPs.
\begin{lemma}(Lemma 4.1 in \cite{khajehnejad2012reweighted})
\label{lemma.FCPl1bound}
Let $\codebook$ be a code that has the $FCP(\F,\C)$ for some index set $\F$ and some $\C\geq1$. Suppose that a codeword $\bfx$ from $\codebook$ is transmitted through a bit flipping channel, and the received codeword is $\bfy$. If the pseudocodeword $\bfxpseudo$ is the output of the LP decoder for the received codeword $\bfy$, then the following holds:
\begin{equation*}
\Vert\bfxpseudo - \bfx\Vert_1 < 2 \frac{\C+1}{\C-1}\Vert(\bfy - \bfx)_{\F^c}\Vert_1.
\end{equation*}
\end{lemma}

Using Lemma.~\ref{lemma.FCPl1bound}, we show our main theorem for two-round \pRLP.
\begin{theorem}
\label{thm.reweightguarantee}
Consider bit flipping channels. If a code $\codebook$ satisfies $FCP(p_{sd}n,\C)$ for $\C>1$ and $p_{sd} > 0$ then for all $0<\epsilon  < \frac{\C-1}{4(\C+1)}$, every error set of size $(1+\epsilon)p_{sd}n$ can be corrected by Algorithm~\ref{Algorithm:reweightedLP}.
\end{theorem}
\begin{proof}
Assume that the all-zeros codeword is transmitted. Suppose the received
vector $\bfy$ has $(1+\epsilon)p_{sd}n$ errors. Let
$\F\subset\{1,2,\dots,n\}$ be the set of these flipped bits.  By
Lemma~\ref{lemma.FCPl1bound},
\begin{equation*}
\Vert\bfx^{(p)}\Vert_1<2\frac{\C+1}{\C-1}\epsilon p_{sd}n.
\end{equation*}
Let $\mcS:=\{i|\xpseudo_i<0.5\}$, then 
\begin{equation}
\label{eq.set_bound}
\left\vert \mcS^c\right\vert < 4\frac{\C+1}{\C-1}\epsilon p_{sd}n.
\end{equation}
Let $\epsilon < \frac{\C-1}{4(\C+1)}$, then $\left\vert \mcS^c\right\vert < p_{sd} n$. Then $\codebook$ has $FCP(\mcS^c,\C)$.
Now rewrite $\bfgamma^{\prime}$ as 
\begin{equation}
\label{eq.reweightedllrassignment}
\gamma_i^\prime = 
\left\{ \begin{array}{rcl}
1+\alpha,&\qquad &i\in \mcS\cap \F^c\\
1-\alpha,&\qquad &i\in \mcS^c\cap \F^c\\
-1-\alpha,&\qquad &i\in \mcS^c\cap \F\\
-1+\alpha,&\qquad &i\in \mcS\cap \F
\end{array}\right..
\end{equation}
Let $\alpha > 1$, then $\alpha - 1<0$. Therefore $\forall \bfw\in\fundcone$,
\begin{align*}
\sum_i\gamma_i^\prime \w_i &= \sum_{i\in \mcS\cap \F^c} (1+\alpha)\w_i + \sum_{i\in \mcS^c\cap \F^c} (1-\alpha)\w_i\\
&\quad+\sum_{i\in \mcS^c\cap \F} (-1-\alpha)\w_i + \sum_{i\in \mcS\cap \F}(-1+\alpha)\w_i \\
&> \sum_{i\in \mcS\cap \F^c} (-1+\alpha)\w_i + \sum_{i\in \mcS\cap \F}(-1+\alpha)\w_i\\
&\quad+\sum_{i\in \mcS^c\cap \F} (-1-\alpha)\w_i + \sum_{i\in \mcS^c\cap \F^c} (-1-\alpha)\w_i\\
&= \sum_{i\in \mcS} (-1+\alpha)\w_i + \sum_{i\in \mcS^c}(-1-\alpha)w_i\\
&= (\alpha-1)\left[ \sum_{i\in \mcS}\w_i - \frac{\alpha+1}{\alpha-1}\sum_{i\in \mcS^c}\w_i \right]\\
&= (\alpha-1)\left[\Vert \w_\mcS\Vert_1 - \frac{\alpha+1}{\alpha-1}\Vert \w_{\mcS^c}\Vert_1\right].
\end{align*}
Since $\codebook$ has $FCP(\mcS^c,\C)$, pick $\alpha > \frac{\C+1}{\C-1}$, then $1<\frac{\alpha+1}{\alpha-1} < \C$ and
\begin{equation*}
\Vert \w_\mcS\Vert_1 - \frac{\alpha+1}{\alpha-1}\Vert \w_{\mcS^c}\Vert_1 > \Vert \w_\mcS\Vert_1 - \C\Vert \w_{\mcS^c}\Vert_1 > 0.
\end{equation*}
Therefore $\sum_i\gamma_i^\prime w_i > 0 $ for all $ \bfw\in\fundcone$. This implies that the all-zeros solution is the minimizer. Then LP problem~(\ref{eq.reweightedlp}) recovers the all-zeros codeword.
\end{proof}
\begin{corollary}
Let $\mathcal{G}$ be the factor graph of a code $\codebook$ of length $n$ and the rate $R = \frac{m}{n}$, and let $\delta>2/3+1/c$. If $\mathcal{G}$ is a bipartite $(\alpha n , \delta c)$ expander graph, then Algorithm~\ref{Algorithm:reweightedLP} succeeds, as long as at most $t(1+\epsilon)n$ bits are flipped by the channel, where $t = \frac{3\delta - 2}{2\delta - 1}\alpha$ and $\epsilon = \frac{1}{16\delta c - 8c - 4}$.
\end{corollary}
\begin{proof}
By Theorem 5.1 in \cite{khajehnejad2012reweighted}, $FCP(tn,\C)$ holds, where $\C = \frac{2\delta - 1}{2\delta - 1 - 1/c}$. Then by Theorem~\ref{thm.reweightguarantee}, we can correct $(1+\epsilon)tn$ errors.
\end{proof}

\subsection{A special case of two-round \pRLP}
\label{subsec.rlpinf}
We now examine the reweighted LP obtained by letting $\alpha\rightarrow\infty$. The algorithm is summarized in Algorithm~\ref{Algorithm:another_reweightedLP}. We can consider the linear program~\eqref{eq.newrlp} as the LP decoding under a new bit flipping channel, where the output of the channel is determined solely by penalizing pseudocodewords.
\begin{algorithm}
\caption{Two-round \pRLP, $\alpha = +\infty$}
\label{Algorithm:another_reweightedLP}
\begin{algorithmic}[1]
\STATE Solve \pLP.
\IF {The solution $\bfx^{(0)}$ is an integral solution}
    \STATE  Stop. {\bf return} $\bfx^{(0)}$.
\ELSE
    \STATE  Save the pseudocodeword $\bfxpseudo = \bfx^{(0)}$ and proceed.
\ENDIF
\FORALL{ $ i \in \mathcal{I} $}
\STATE Update $\gamma'_i = - \sgn(x^{(p)}_i - 0.5)$.
\ENDFOR
\STATE Solve
\begin{equation}
\label{eq.newrlp}
\begin{split}
\min & \quad\bfgamma'^{T} \bfx \\
\suchthat & \quad\bfP_j \bfx \in \PP_{d_j} .
\end{split}
\end{equation}
\end{algorithmic}
\end{algorithm}

\begin{corollary}
\label{corollary.newreweigthedlp}
Theorem~\ref{thm.reweightguarantee} also holds for Algorithm~\ref{Algorithm:another_reweightedLP}.
\end{corollary}
\begin{proof}
The only difference between this proof and proof for Theorem~\ref{thm.reweightguarantee} is the values for $\bfgamma^\prime$ in \eqref{eq.reweightedllrassignment}. The new values are
\begin{equation*}
\gamma_i^\prime = 
\left\{ \begin{array}{rcl}
\alpha,&\qquad &i\in T\\
-\alpha,&\qquad &i\in T^c
\end{array}\right..
\end{equation*}
Then we can normalize the objective and let $\alpha = \pm1$. Using $FCP$, we immediately deduce that this LP recovers the all-zero codeword.
\end{proof}

Note that Theorem~\ref{thm.reweightguarantee} is conservative since
the proof considers the worst case scenario. Probabilistic analyses
introduced in~\cite{daskalakis2008probabilistic}
and~\cite{arora2012message} cannot be applied directly to this case
because the second LP takes the pseudocodeword that results from the
first LP as part of its input. The probabilistic properties for
pseudocodewords should not be considered to be the same as the
received vectors from a memoryless channel. Thus, probabilistic
analyses for reweighted LP decoding remains open.

\subsection{Numerical results of \pRLP}
\label{sec.numerical.reweightedlp}
In this section, we show numerical results for error performance of \pRLP. In Fig.~\ref{fig.sim_margulis_rlp}, we again simulate the $[2640,1320]$ Margulis code using two-round \pRLP and compare its performance to other decoders. We make the following 
observations: First, two-round \pRLP outperforms LP decoding significantly (by around $0.3\dB$) even though the provable 
improvement presented in the previous section is small. This shows that \pRLP is a good practical decoding algorithm for decoding LDPC codes. Second, there is still a $0.3\dB$ SNR gap between
\pPLP and two-round \pRLP. Since two-round \pRLP solves two LPs, each with complexity similar to \pPLP, it is still desirable to 
apply \pPLP in practice. Interested readers are referred to~\cite{liu2012the} for simulation results for the BSC.
\begin{figure}[h]
\psfrag{&ADMMLP}{\scalebox{.6}{ADMM LP decoding}}
\psfrag{&ADMM-ReweightedLP}{\scalebox{.6}{Two-round \pRLP}}
\psfrag{&BPNonSat}{\scalebox{.6}{Non-saturating BP decoding}}
\psfrag{&&ADMM-L1PenaliedDecoder}{\scalebox{.6}{\pPD, $\ell_1$ penalty }}
\psfrag{&G}{\scalebox{.6}{\hspace{-2mm}$E_b/N_0(\dB)$}}
\psfrag{&H}{\scalebox{.6}{WER}}

\psfrag{1}{\scriptsize{$1$}}
\psfrag{1.5}{\scriptsize{$1.5$}}
\psfrag{2}{\scriptsize{$2$}}
\psfrag{2.5}{\scriptsize{$2.5$}}
\psfrag{3}{\scriptsize{$3$}}

\psfrag{10}{\scriptsize{\hspace{-1mm}$10$}}
\psfrag{-2}{\tiny{\hspace{-1mm}$-2$}}
\psfrag{-4}{\tiny{\hspace{-1mm}$-4$}}
\psfrag{-6}{\tiny{\hspace{-1mm}$-6$}}
\psfrag{-8}{\tiny{\hspace{-1mm}$-8$}}

\psfrag{e0}{\tiny{\hspace{-1mm}$0$}}

    \begin{center}
    \includegraphics[width=3.5in]{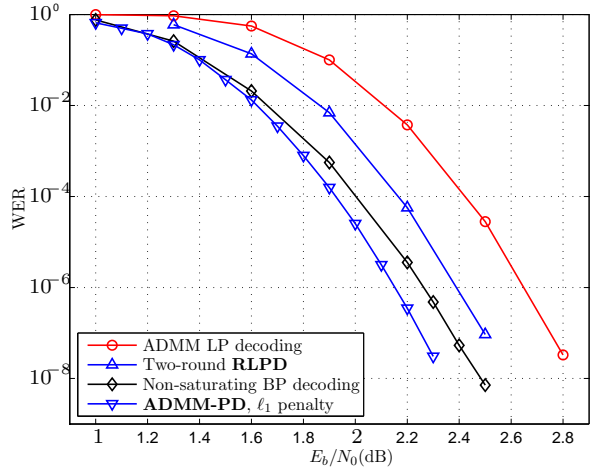}
    \end{center}
    \caption{Word-error-rate (WER) of the
      $[2640,1320]$ ``Margulis'' code plotted as a function of $E_b / N_0$ for the AWGN channel.}
    \label{fig.sim_margulis_rlp}
\end{figure}

We also study the impact of allowing multiple rounds for \pRLP. In Fig.~\ref{fig.sim_margulis_AWGN_multiple_rounds}, we plot the WER performance as a function of the number of rounds allowed for \pRLP. From this experiment, we conclude that by allowing multiple rounds for \pRLP, we can achieve better error performance. However the gain from every additional round of reweighting decreases as the number of rounds increases. But even with 10 rounds, the \pRLP cannot surpass BP decoding or \pPD, making it less favorable in scenarios where provable guarantees are not necessary.
\begin{figure}
\psfrag{&LPDecoder}{\scalebox{.6}{\pLP}}
\psfrag{&MultipleRoundsRLP}{\scalebox{.6}{\pRLP, $\alpha = 0.6$}}
\psfrag{&NonsatBP}{\scalebox{.6}{Non-saturating BP}}
\psfrag{&L1PD}{\scalebox{.6}{\pPLP, $\alpha = 0.6$}}

\psfrag{&G}{\scalebox{.6}{\hspace{-1cm} \# of rounds for \pRLP}}
\psfrag{&H}{\scalebox{.6}{\hspace{-1cm} word-error-rate (WER)}}
    \begin{center}
    \includegraphics[width=3.5in]{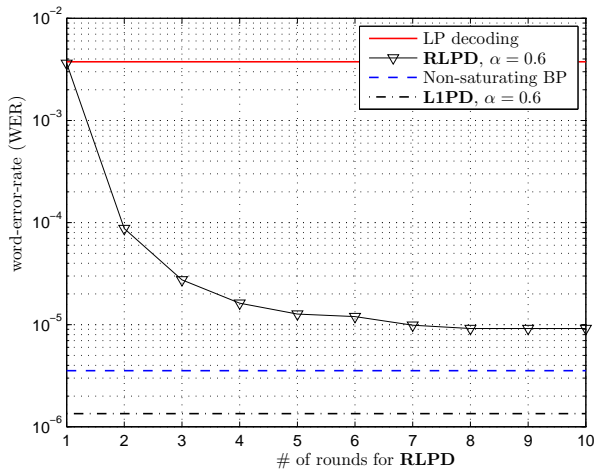}
    \end{center}
    \caption{Effects of multiple-round reweighting. The word-error-rate (WER) of the $[2640,1320]$ ``Margulis'' code for the AWGN channel is plotted as a function of \# of rounds of reweightings. $E_b/N_0 = 2.2\dB$.}
    \label{fig.sim_margulis_AWGN_multiple_rounds}
\end{figure}
\section{Conclusions}
\label{sec.conclusion}
In this paper, we introduce a class of high performing decoders for LDPC codes. The key idea is the
addition of a penalty term to the linear objective of LP
decoding with the intent of suppressing pseudocodewords.

We show that by constructing an objective that includes an 
$\ell_1$ or an $\ell_2$ penalty we are
able to decode LDPC codes in a way that combines the best
characteristics of BP and LP decoders. We summarize the 
advantages of the ADMM penalized decoders. First,
in our simulations, the decoders associated with both $\ell_1$ and $\ell_2$ penalties close the SNR gap
between LP and BP decoding. They either achieve (for the MacKay code)
or outperform (for the Margulis code) BP at low SNRs. Second, the ADMM penalized decoder has similar WER slope 
as LP decoding at high SNRs, steeper than BP decoding which suffers from an error floor. 
Third, the ADMM penalized decoder attains good error performance with no sacrifices in decoding speed. On the contrary, it outperforms both BP and LP decoding in terms of decoding speed in our software serial implementation. All these aspects make the ADMM penalized decoder a strong competitor in practice to both BP and LP decoding.

The objective function of the decoding problem is
non-convex. To analyze this non-convex decoder, we use the concept of 
instantons. We extend the instanton search algorithms in the literature to
the ADMM penalized decoder. Using the instanton search algorithm, we 
show that we can accurately predict the high SNR behavior of the ADMM penalized decoder.
In order to regain provable error correction guarantees similar to the LP decoder, we linearize the penalized objective thereby bringing the problem to the convex realm. When applied to bit flipping channels,
we show that this linearized problem has better theoretical guarantees 
when compared to the LP decoder. In addition, we show through simulation results
that it significantly outperforms LP decoding.

We identify several future directions of this decoder. We show in this paper that
decoding failures are, again, related to trapping sets. Therefore understanding the dynamics
of ADMM on trapping sets can be useful for better
predicting the high SNR performance of the decoder. This is left for future work.
Moreover, even for non-convex programs, it may be possible to prove convergence
under some special conditions. This question is open and is left for future work. 
Since this decoder achieves the best performance compared to both LP and BP decoding, it 
is natural to ask whether it can be efficiently implemented in hardware. This makes for
a third future direction. 
\section*{Appendix}
\section{Proof of theorem~\ref{thm.cw_independent} and corollary~\ref{crly.initialpoints}}
\subsection{Proof of theorem~\ref{thm.cw_independent}}
\label{appendix.cwindep_proof}
We first sketch the proof without proving a key lemma: Lemma~\ref{lemma.mapping}. Then we prove Lemma~\ref{lemma.mapping}.
\begin{proof}
We need to prove $\Pr[\error|0^n] = \Pr[\error|\bfc]$, where $\bfc$ is any non-zero codeword. Let
\begin{equation*}
\mcB(\bfc):=\{\bfy|\text{Decoder fails to recover }\bfc\text{ if }\bfy \text{ is received}\}
\end{equation*}
Then $\Pr[\error|\bfc] = \sum_{y\in \mcB(c)}\Pr[\bfy|\bfc]$. We show in Lemma~\ref{lemma.mapping} that for any codeword $\bfc$, there exists a one-to-one mapping from the received vector $\bfy$ to a vector $\bfy^0$, such that
the following two statements hold:
\begin{enumerate}
\item $\Pr[\bfy|\bfc] = \Pr[\bfy^0|0^n]$,
\item $\bfy \in \mcB(\bfc)$ if and only if $\bfy^0\in \mcB(0^n)$.
\end{enumerate}
Then 
\begin{align*}
\Pr[\error|\bfc] &= \sum_{\bfy\in \mcB(\bfc)}\Pr[\bfy|\bfc]\\
&= \sum_{\bfy^0\in \mcB(0^n)}\Pr[\bfy^0|\bfc] = \Pr[\error|0^n] 
\end{align*}
Lemma~\ref{lemma.mapping} explicitly describes the mapping that satisfies both two statements above and completes the proof.
\end{proof}

\begin{lemma}
\label{lemma.mapping}
Let $\bfc$, $\bfy$  and $\mcB(\bfc)$ be defined the same way as above. Let $\mcZ:\Sigma^n\mapsto\Sigma^n$ be a mapping and $\bfy^0 := \mcZ(\bfy)$ is defined by
\begin{equation*}
y^0_i = 
\begin{cases}
y_i  & \text{if $c_i = 0$,}\\
y_i^\prime &\text{if $c_i = 1$.}
\end{cases}
\end{equation*}
where $y_i^\prime$ is the symmetric symbol of $y_i$ with respect to the symmetric channel\footnote{Following the definition of a symmetric channel in~\cite[p.~178]{richardson2008modern}, $y_i^\prime = -y_i$ in general. For example in the AWGN channel with BPSK modulation, $y_i^\prime = -y_i$. However in the BSC, $y_i^\prime = 0$ if $y_i = 1$ and vice versa.}  $\Pr_{Y|X}$. Then 
\begin{enumerate}
\item $\Pr[\bfy|\bfc] = \Pr[\bfy^0|0^n]$,
\item $\bfy \in \mcB(\bfc)$ if and only if $\bfy^0\in \mcB(0^n)$.
\end{enumerate}
\end{lemma}
\begin{proof}
\begin{itemize}
\item Proof of statement (1):
\begin{align*}
\Pr[\bfy|\bfc] &= \prod_{i}\Pr[y_i|c_i]\\
&=\prod_{i:c_i = 0}\Pr[y_i|0]\prod_{i:c_i = 1}\Pr[y_i|1]\\
&=\prod_{i:c_i = 0}\Pr[y^0_i|0]\prod_{i:c_i = 1}\Pr[y_i^\prime|0]\\
&=\prod_{i}\Pr[y^0_i|0] = \Pr[\bfy^0|0^n]
\end{align*}

\item Proof of statement (2): We first define \textit{relative vector}. The definition here is slightly different from the definition in~\cite{feldman2005using}. 
\begin{definition}
The \textit{relative vector} of a vector $\bfa\in\mathbb{R}^\blocklength$ with respect to codeword $\bfc$ of length $\blocklength$ is denoted by $\bfa^r:=\mcR_{\bfc}(\bfa)$, which is obtained as
\begin{equation*}
\mcR_{\bfc}(\bfa)_i = 
\begin{cases}
a_i  & \text{if $c_i = 0$,}\\
1 - a_i &\text{if $c_i = 1$.}
\end{cases}
\end{equation*}
\end{definition}
One special yet useful case is $\mcR_{\bfc}(\bfc) = 0^n$. 

With a small abuse of notation, we allow operator $\mcR$ take replicas of ADMM formulation $\bfz_j$ as input. Since the entries in $\bfz_j$ are inherently associated with indices of the codeword $\bfc$, $\mcR_{\bfc}(\bfz_j)$ defines the relative vector of $\bfz_j$ with respect to the corresponding sub-vector (defined by the $j$-th check) of $\bfc$.

Now let $\hat{\bfc} := \mathscr{D}(\bfy)$ be the output of the decoder if $\bfy$ is received and $\hat{\bfc}^0 := \mathscr{D}(\bfy^0)$. We show in Lemma~\ref{lemma.equaldecoding} that $\hat{\bfc}^0 = \mcR_{\bfc}(\hat{\bfc})$. Then $\mathscr{D}(\bfy) = \bfc$ if and only if $\mathscr{D}(\bfy^0) =\mcR_{\bfc}(\bfc) = 0^n$, which implies the second statement. 
\end{itemize}
\end{proof}
\begin{lemma}
\label{lemma.quiviter}
In Algorithm~\ref{Algorithm:ADMM}, let  $\bfx^{(k)}$, $\bfz_j^{(k)}$ and $\bflambda_j^{(k)}$ be the vectors after the $k$-th iteration when decoding $\bfy$. Let $\bfx^{0,(k)}$, $\bfz_j^{0,(k)}$ and $\bflambda_j^{0,(k)}$ be the vectors after the $k$-th iteration when decoding $\bfy^0$. If $\bfx^{(k)} = \mcR_{\bfc}(\bfx^{0,(k)})$,  $\bfz_j^{(k)} = \mcR_{\bfc}(\bfz_j^{0,(k)})$ and $\bflambda_j^{(k)} = -\bflambda_j^{0,(k)}$ then $\bfx^{(k+1)} = \mcR_{\bfc}(\bfx^{0,(k+1)})$,  $\bfz_j^{(k+1)} = \mcR_{\bfc}(\bfz_j^{0,(k+1)})$ and $\bflambda_j^{(k+1)} = -\bflambda_j^{0,(k+1)}$.
\end{lemma}
\begin{proof} 
We drop the iterate $(k)$ for simplicity and denote $\bfx^{\new}$ , $\bfz_j^{\new}$ and $\bflambda^{\new}_j$ to be the updated vector at $(k+1)$-th iteration. 
Let $\bfx = \mcR_{\bfc}(\bfx^{0})$,  $\bfz_j = \mcR_{\bfc}(\bfz_j^{0})$ and $\bflambda_j = -\bflambda_j^{0}$. Also let $\bfgamma^0$ be the log-likelihood ratio for the received vector $\bfy^0$.
At the $\bfx$-update, it is sufficient to verify indices $i$ where $c_i = 1$. From Algorithm~\ref{Algorithm:ADMM}, $x_i^{\new}$ is the root with the maximum distance from $0.5$ of the following equation
\begin{equation}
\label{eq.x_update_original}
x = \frac{1}{\degi}\left(\sum_j \left(\bfz_j^{(i)} - \frac{\lmbpara_j^{(i)}}{\penpara}\right) - \frac{1}{\penpara}(\gamma_i + g^\prime(x))\right).
\end{equation}
By the symmetry property of the penalty function $g$, $1-x_i^{\new}$ is the root of equation
\begin{equation}
\label{eq.x_update_relative}
x = \frac{1}{\degi}\left(\sum_j\left(1 - \bfz_j^{(i)} + \frac{\lmbpara_j^{(i)}}{\penpara}\right) + \frac{1}{\penpara}(\gamma_i - g^\prime(x))\right).
\end{equation}
Since the channel is symmetric, $\bfgamma = - \bfgamma^0$. Then~\eqref{eq.x_update_relative} can be rewritten as 
\begin{equation}
\label{eq.x_update_zero}
x = \frac{1}{\degi}\left(\sum_j\left(\bfz_j^{0,(i)} - \frac{\lmbpara_j^{0,(i)}}{\penpara}\right) - \frac{1}{\penpara}(\gamma_i^{0} + g^\prime(x))\right).
\end{equation}
Suppose there exists another root $t$ for equation~\eqref{eq.x_update_zero} such that $|t - 0.5| > |(1-x_i^{\new}) - 0.5|$ then $1-t$ is a root for equation~\eqref{eq.x_update_original}, contradicting with the fact that $x_i^{\new}$ is the root of equation~\eqref{eq.x_update_original} that has \textbf{the largest} distance from 0.5. Therefore 
\begin{equation*}
x_i^{\new} = 1 - x_i^{\new,0}
\end{equation*}
Let $\bfv_j  = \bfP_j \bfx + \bflambda_j /\mu$ and $\bfv^0_j  = \bfP_j \bfx^0 + \bflambda^0_j /\mu$, then $\bfv^{0,(i)}_j = 1 - (\bfP_j \bfx)^{(i)} - \bflambda^{(i)}_j /\mu = 1 - \bfv^{(i)}_j$. Now $ \bfz^{\new}_j = \Proj_{\PP_d} (\bfv_j) $ and $ \bfz^{0,\new}_j = \Proj_{\PP_d} (\bfv^0_j) $. This means that we need to show $ \bfz^{\new}_j = \mcR_{\bfc}(\bfz^{0,\new}_j)$. In order to show this, we use Lemma 17 in~\cite{feldman2005using}. 

Suppose the the opposite is true, that is $\bfz^{0,\new}_j$ is the projection of $\bfv^0_j$ onto the parity polytope, but $\mcR_{\bfc}(\bfz^{0,\new}_j)$ is not the projection of $\bfv_j$. By Lemma 17 in~\cite{feldman2005using}, $\mcR_{\bfc}(\bfz^{0,\new}_j)$ is inside the parity polytope. Suppose the real projection is $\bfz_j^\prime$, then $\mcR_{\bfc}(\bfz_j^\prime)$ is also in the parity polytope and 
\begin{align*}
\Vert\mcR_{\bfc}(\bfz_j^\prime) -\bfv^0_j \Vert_2 &= \Vert\bfz_j^\prime -\bfv_j \Vert_2 \\
&< \Vert\mcR_{\bfc}(\bfz^{0,\new}_j) -\bfv_j \Vert_2 =\Vert\bfz^{0,\new}_j -\bfv^0_j \Vert_2
\end{align*}

Now $\mcR_{\bfc}(\bfz_j^\prime)$ is closer to $\bfv_j^0$ than $\bfz^{0,\new}_j$ is. This contradicts the assumption that $\bfz^{0,\new}_j$ is the projection.

It remains to verify one more equality:
\begin{align*}
\bflambda^{\new,(i)}_j &=  \bflambda^{(i)}_j + \mu \left( (\bfP_j \bfx^{\new})^{(i)} - \bfz_j^{\new,(i)}\right) \\
&= -\bflambda^{0,(i)}_j + \mu \left(\left(1 -  (\bfP_j \bfx^{0,\new})^{(i)}\right)\right.\\
 &\quad -\left.\left( 1-\bfz_j^{0,\new,(i)}\right)\right)\\
&= -\left(\bflambda^{0,(i)}_j + \mu \left( (\bfP_j \bfx^{0,\new})^{(i)} - \bfz_j^{0,\new,(i)}\right) \right)\\
&= -\bflambda^{0,\new,(i)}_j
\end{align*}
\end{proof}
\begin{lemma}
\label{lemma.equaldecoding}
Let $\hat{\bfc} = \mathscr{D}(\bfy)$ be the output of the decoder if $\bfy$ is received and $\hat{\bfc}^0 = \mathscr{D}(\bfy^0)$. Then $\hat{\bfc}^0 = \mcR_{\bfc}(\hat{\bfc})$.
\end{lemma}
\begin{proof}
We drop the iteration number $k$ for simplicity. Note that for all $j$, $\bflambda_j$ and $\bfz_j$ are initialized such that  $\bfz_j = \mcR_{\bfc}(\bfz_j^0) = 0.5$ and $\bflambda_j = -\bflambda_j^0 = 0$. By induction on Lemma~\ref{lemma.quiviter}, Lemma~\ref{lemma.quiviter} holds for all iterations. It remains to prove that both decoding processes stop at the same iteration. When determining the termination criterion, 
\begin{align*}
 \left\vert\left(\bfP_j \bfx \right)_i - \left(\bfz_j\right)_i\right\vert &=  \left\vert\left(1 - \left(\bfP_j \bfx \right)_i\right) - \left(1 - \left(\bfz_j\right)_i\right)\right\vert\\
&=\left\vert \left(\bfP_j \bfx \right)^r_i - \left(\bfz_j\right)^r_i\right\vert = \left\vert \left(\bfP_j \bfx^0 \right)_i - \left(\bfz^0_j\right)_i\right\vert
\end{align*}
Therefore $$\sum_j { \| \bfP_j \bfx - \bfz_j \|^2_{2} } = \sum_j { \| \bfP_j \bfx^0 - \bfz_j^0 \|^2_{2} }$$ and $$\sum_j { \| \bfz^{k}_j - \bfz^{k-1}_j \|^2_{2} } = \sum_j { \| \bfz^{k,0}_j - \bfz^{k-1,0}_j \|^2_{2} } $$ for every iteration. This implies decoding $\bfy$ and $\bfy^0$ will terminate at the same iterate. The output of the decoder, $\mathscr{D}(\bfy)$ and $\mathscr{D}(\bfy^0)$, will then be relative vectors according to Lemma~\ref{lemma.quiviter}. Therefore $\hat{\bfc}^0 = \mcR_{\bfc}(\hat{\bfc})$.
\end{proof}

\subsection{Proof of corollary~\ref{crly.initialpoints}}
\label{appendix.cwindep_proof_corollary}
\begin{enumerate}
\item Proof for the first case.

Lemma~\ref{lemma.equaldecoding} still holds if the initialized vectors $\bfz_j^{(0)} = \mcR_{\bfc}(\bfz_j^{0,(0)})$ and $\bflambda_j^{(0)} = -\bflambda_j^{0,(0)} = 0$. For any initial vector for decoding $\bfy^0$, there exist a initial vector for decoding $\bfy$ that has the same decoding failure probability. Since $\bfz_j$ is initialized uniformly, the total probability of error is the same.

\item Proof for the second case.

\begin{lemma}
\label{lemma.lprelative}
In LP decoding~\eqref{eq.lpfeldman}, let $\hat{\bfc} = \mathscr{D}(\bfy)$ be the output of the decoder if $\bfy$ is received and $\hat{\bfc}^0 = \mathscr{D}(\bfy^0)$. Then $\hat{\bfc}^0 = \mcR_{\bfc}(\hat{\bfc})$.
\end{lemma}
\begin{proof} Let $\hat{\bfc}_r := \mcR_{\bfc}(\hat{\bfc})$.
By Lemma 18 in~\cite{feldman2005using}, 
\begin{equation*}
\sum_i\gamma_i\hat{c}_i - \sum_i\gamma_iy_i = \sum_i\gamma_i^0\hat{c}_{r,i} - \sum_i \gamma_i^00.
\end{equation*}
Suppose that the opposite is true, that is $\hat{\bfc}^0 \neq \hat{\bfc}_r$, then
\begin{equation*}
\sum_i\gamma_i^0\hat{c}^0_{i}<\sum_i\gamma_i^0\hat{c}_{r,i} = \sum_i\gamma_i\hat{c}_i - \sum_i\gamma_iy_i.
\end{equation*}
Let $\hat{\bfc}^0_r := \mcR_{\bfc}(\hat{\bfc}^0)$. Then $\hat{\bfc}^0_r$ is feasible by Lemma 17 in~\cite{feldman2005using}. Then
\begin{align*}
\sum_i\gamma_i^0\hat{c}^0_{r,i}-\sum_i\gamma_iy_i &= \sum_i\gamma_i^0\hat{c}^0_{i} - \sum_i \gamma_i^00\\
&<\sum_i\gamma_i\hat{c}_i - \sum_i\gamma_iy_i,
\end{align*}
which implies that 
\begin{align*}
\sum_i\gamma_i^0\hat{c}^0_{r,i}&<\sum_i\gamma_i\hat{c}_i
\end{align*}
This means that $\hat{\bfc}$ is not the minimizer, contradicting the fact that it is the output of the LP decoder.
\end{proof}

With lemma~\ref{lemma.lprelative}, the initialized vectors satisfies $\bfz_j^{(0)} = \mcR_{\bfc}(\bfz_j^{0,(0)})$ and $\bflambda_j^{(0)} = -\bflambda_j^{0,(0)} = 0$. Thus by lemma~\ref{lemma.equaldecoding}, Theorem~\ref{thm.cw_independent} still holds.
\end{enumerate}

\section{Proof of Theorem~\ref{thm.reweightedlpcwindependence}}
\label{appendix.cwindep_proof_reweightedlp}
\begin{proof}
There are two cases: (i) the decoder outputs an integral solution in the first LP and terminates or (ii) the decoder proceeds to the reweighted LP step. For the first case, the decoding error probability is independent of the transmitted codeword by standard LP decoding results~\cite{feldman2005using}. For the second case, we only need to show that $\bfgamma^\prime = -\bfgamma^{\prime,0}$. Then by Lemma~\ref{lemma.lprelative}, we still get relative solutions, and therefore are equivalent. Let $\bfxpseudo_r := \mcR_{\bfc}(\bfxpseudo)$, then
\begin{align*}
\gamma^\prime_i &= \gamma_i + g^\prime(\xpseudo_i)\\
&= -\gamma_i^0 - g^\prime(\xpseudo_{r,i})\\
&= -\gamma_i^{\prime,0}
\end{align*}
where the second equality comes from that $g$ is symmetric at $x = 0.5$, so that $g^\prime(x) = - g^\prime(1-x)$.
\end{proof}

\end{document}